\documentclass[submission,copyright,creativecommons]{eptcs}
\providecommand{\event}{TLLA/Linearity 2018} 
\usepackage{breakurl}             
\usepackage{underscore}           
\usepackage{etex}
 \usepackage{appendix}
\usepackage{stackrel}
\usepackage{rotating}
 \usepackage[T1]{fontenc}
\usepackage{graphicx}
\usepackage{tabularx}
\usepackage{amsfonts}
\usepackage{cmll}
\usepackage{amsthm}
\usepackage{amsmath}
\usepackage{proof}
\usepackage{mathdots}
\usepackage{wrapfig}
\usepackage{xcolor}
\usepackage{amssymb}
\usepackage{color}
\usepackage{pigpen}
\usepackage{adjustbox}
\usepackage{tikz}
\usepackage{mathrsfs}
\usepackage[all]{xy}
\usepackage{lscape}
\usepackage{stmaryrd}
\usepackage{mathdots}
\usepackage{bbm}

\usepackage{bussproofs}
\usepackage{subcaption}

\EnableBpAbbreviations

\usetikzlibrary{cd}

\newtheorem{example}{Example}
\newtheorem{definition}{Definition}

\newtheorem{remark}{Remark}
\newtheorem{theorem}{Theorem}

\newtheorem{lemma}{Lemma}

\newtheorem{proposition}{Proposition}

\numberwithin{equation}{section}



\newtheorem{corollary}[theorem]{Corollary}

\newcommand{\B}[1]{\mathbf{#1}}
\newcommand{\TT}[1]{\mathtt{#1}}
\newcommand{\D}[1]{\mathscr{#1}}

\newcommand{\C}[1]{\mathcal{#1}}
\newcommand{\BB}[1]{\mathbb{#1}}

\newcommand{\OV}[1]{\overline{#1}}

\newcommand{\To}{\Rightarrow}

\newcommand{\po}{\ar@{}[dr]|{\text{\pigpenfont R}}}
\newcommand{\pb}{\ar@{}[dr]|{\text{\pigpenfont J}}}

\definecolor{color0}{HTML}{4682B4}

\title{Proof Nets, Coends and the Yoneda Isomorphism}
\author{Paolo Pistone
\institute{Wilhelm-Schickard-Institut \\
Eberhard Karls Universit\"at T\"ubingen}
\email{paolo.pistone@uni-tuebingen.de}}

\def\titlerunning{Proof Nets, Coends and the Yoneda Isomorphism}
\def\authorrunning{Paolo Pistone}
\begin{document}
\maketitle

\begin{abstract}
Proof nets provide permutation-independent representations of proofs and are used to investigate coherence problems for  monoidal categories. We investigate a coherence problem concerning Second Order Multiplicative Linear Logic ($\mathsf{MLL2}$), that is, the one of characterizing the equivalence over proofs generated by the interpretation of quantifiers by means of ends and coends.

We provide a compact representation of proof nets for a fragment of $\mathsf{MLL2}$ related to the Yoneda isomorphism. By adapting the ``rewiring approach'' used in coherence results for $^{*}$-autonomous categories, we define an equivalence relation over proof nets called ``rewitnessing''. We prove that this relation characterizes, in this fragment, the equivalence generated by coends.

\end{abstract}

\section{Introduction}

Proof nets are usually investigated as canonical representations of proofs. For the proof-theorist, the adjective ``canonical'' indicates a representation of proofs insensitive to admissible permutations of rules; for the category-theorist, it indicates a faithful representation of arrows in free monoidal categories (e.g.\ $^{*}$-autonomous categories), by which coherence results can be obtained.

This twofold approach has been developed extensively in the case of Multiplicative Linear Logic (see for instance \cite{Blute1993,Blute1996}). The use of $\mathsf{MLL}$ proof nets to investigate coherence problems relies on the correspondence between proof nets and a particular class of dinatural transformations (see \cite{Blute1993}). As dinatural transformations provide a well-known interpretation of parametric polymorphism (see \cite{Bainbridge1990,Girard1992}), it is natural to consider the extension of this correspondence to second order Multiplicative Linear Logic $\mathsf{MLL2}$. This means investigating the ``coherence problem'' generated by the interpretation of quantifiers as ends/coends, that is, to look for a faithful proof net representation of coends over a $^{*}$-autonomous category.

The main difficulty of this extension is that, as is well-known, dinaturality does not scale to second order (e.g. System $\mathsf F$, see \cite{Delatail2009}): the dinatural interpretation of proofs generates an equivalence over proofs which strictly extends the equivalence generated by $\beta$ and $\eta$ conversions. 
In particular, coends induce ``generalized permutations'' of rules (\cite{StudiaLogica}) to which neither System $\mathsf F$ proofs nor standard proof nets for $\mathsf{MLL2}$ are insensitive. For instance, the interpretation of quantifiers as ends/coends (whose definition is recalled in appendix \ref{appB}) equates the distinct System $\mathsf F$ derivations in fig. \ref{coend1} as well as the distinct proof nets in fig. \ref{nocoend}. From these examples it can be seen that such generalized permutations do not preserve the witnesses of existential quantification (or, equivalently, of the elimination of universal quantification). 

Several well-known issues in the System $\mathsf F$ representation of categorial structures can be related to this phenomenon. For instance, the failure of universality for the ``Russell-Prawitz'' translation of connectives (e.g.\ the failure of the isomorphism $A\otimes B\simeq \forall X((A\multimap B\multimap X)\multimap X)$), and the failure of initiality for the System $\mathsf F$ representation of initial algebras (i.e.\ the failure of the isomorphism $\mu X.T(X)\simeq \forall X((T(X)\To X)\To X)$). In such cases, the failure is solved by considering proofs modulo the equivalence induced by dinaturality (see \cite{Plotkin1993, Hasegawa2009}). 
All these can be seen as instances of a more general problem, namely the fact that the \emph{Yoneda isomorphism}
$ Nat(\BB C(a,x),F)\simeq F(a)$ corresponds, in the language of $\mathsf{MLL2}$, to a series of logical equivalences of the form
 $\forall X((A\multimap X)\multimap F[X]))\simeq F[A/X]$ which fail to be isomorphisms of types. In this paper we investigate the possibility to provide a faithful representation of the Yoneda isomorphism, and more generally of ends and coends, by means of $\mathsf{MLL2}$ proof nets.

\begin{figure}
\begin{subfigure}{0.48\textwidth}
\adjustbox{scale=0.7,center}{$
\begin{matrix}
\AXC{$\forall X (X\to X)$}
\UIC{$A\to A$}
\AXC{$[A]^{n}$}
\BIC{$A$}
\noLine
\UIC{$f$}
\noLine
\UIC{$B$}
\RL{$n$}
\UIC{$A\to B$}
\DP
& \ \not\simeq \ &
\AXC{$\forall X (X\to X)$}
\UIC{$B\to B$}
\AXC{$[A]^{n}$}
\noLine
\UIC{$f$}
\noLine
\UIC{$B$}
\BIC{$B$}
\RL{$n$}
\UIC{$A\to B$}
\DP \\
 \ & & \  
\end{matrix}
$}
\caption{Failure of dinaturality in System $\mathsf F$}
\label{coend1}
\end{subfigure}
\begin{subfigure}{0.54\textwidth}
\adjustbox{scale=0.7, center}{$
\begin{tikzpicture}[baseline=-2ex]
\node(a) at (0,0) {$A$};
\node(a') at (1,0) {$A^{\bot}$};
\node(t) at (0.5,-0.75) {$\otimes$};
\node(e) at (0.5,-1.5) {$\exists$};
\node(b) at (2,-0.75) {$A^{\bot}$};
\node(b') at (3,-0.75) {$B$};
\node(p) at (2.5,-1.5) {$\parr$};

\draw[thick] (a) to (t);
\draw[thick] (a') to (t);
\draw[thick] (e) to (t);
\draw[thick] (b) to (p);
\draw[thick] (b') to (p);

\node(f) at (2,0.7) [draw,thick,minimum width=0.9cm,minimum height=0.25cm, rounded corners=3pt] {$f$};
\draw[thick] (a) to [bend left=85] (b);
\draw[thick] (a') to [bend left=45] (f);
\draw[thick] (f) to [bend left=45] (b');

\end{tikzpicture}
\qquad\not\simeq\qquad
\begin{tikzpicture}[baseline=-2ex]
\node(a) at (0,0) {$B$};
\node(a') at (1,0) {$B^{\bot}$};
\node(t) at (0.5,-0.75) {$\otimes$};
\node(e) at (0.5,-1.5) {$\exists$};
\node(b) at (2,-0.75) {$A^{\bot}$};
\node(b') at (3,-0.75) {$B$};
\node(p) at (2.5,-1.5) {$\parr$};

\draw[thick] (a) to (t);
\draw[thick] (a') to (t);
\draw[thick] (e) to (t);
\draw[thick] (b) to (p);
\draw[thick] (b') to (p);

\node(f) at (1,1) [draw,thick,minimum width=0.9cm,minimum height=0.25cm, rounded corners=3pt] {$f$};
\draw[thick] (a') to [bend left=85] (b');
\draw[thick] (a) to [bend left=45] (f);
\draw[thick] (f) to [bend left=45] (b);

\end{tikzpicture}
$}
\caption{Failure of dinaturality for proof nets}
\label{nocoend}
\end{subfigure}
\caption{Failure of dinaturality in System $\mathsf F$ and $\mathsf{MLL2}$}
\end{figure}
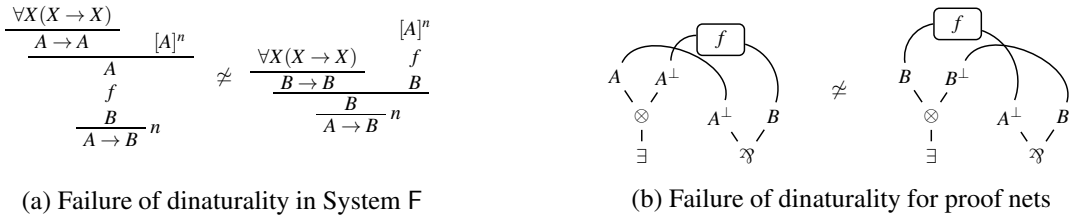


%


As a consequence of the isomorphism $\forall X(X\multimap X)\simeq \B 1$, which is a particular instance of the Yoneda isomorphism just recalled, the proof net representation of quantifiers as ends and coends must include a faithful representation of multiplicative units. From this we can deduce some \emph{a priori} limitations to our enterprise: it is well-known that no canonical representation of $\mathsf{MLL}$ with multiplicative units can have both a tractable correctness criterion and a tractable translation from sequent calculus (\cite{Heij2014}). 
However, in usual approaches to multiplicative units proof nets are considered modulo an equivalence relation called \emph{rewiring} (\cite{Trimble,Blute1996,Hughes2012}), which provides a partial solution to this problem. The ``rewiring approach'' (\cite{Hughes2012}) allows to circumvent the complexity of checking arrows equivalence in the free $^{*}$-autonomous category by isolating the complex part into a geometrically intuitive equivalence relation.

We define a compact representation of proof nets (called $\exists$-linkings) for the fragment of $\mathsf{MLL2}$ which adapts the rewiring technique to second order quantification. We consider the system $\mathsf{MLL2}_{\C Y}$, in which quantification $\forall XA$ is restricted to ``Yoneda formulas'', i.e. formulas of the form $\forall X((\bigotimes_{i}^{n}C_{i}\multimap X)\multimap D[X])$. This fragment contains the multiplicative ``Russell-Prawitz'' formulas as well as the translation of multiplicative units. 
In our approach rewiring is replaced by \emph{rewitnessing}, an equivalence relation which allows to rename the witnesses of existential quantifiers. This approach is related to rewiring in the sense that, when restricted to the second order translation of units, $\exists$-linkings correspond exactly to the ``lax linkings'' in \cite{Hughes2012}.

Our main result (theorem \ref{equiv}) is that the equivalence over proofs generated by coends coincides exactly with the rewitnessing equivalence over $\exists$-linkings. More precisely, we define an equivalence $\simeq_{\varepsilon}$ over standard $\mathsf{MLL2}$ proof nets, where two proof nets are equivalent when their interpretations in any dinatural model coincide, and we show that, within the fragment $\mathsf{MLL2}_{\C Y}$, $\pi \simeq_{\varepsilon}\pi'$ holds iff the associated $\exists$-linkings $\ell_{\pi}$ and $\ell_{\pi'}$ are equivalent up to rewitnessing. To prove this, we construct an isomorphism between the category generated by $\mathsf{MLL2}$ proof nets modulo the equivalence induced by dinaturality and the category generated by $\exists$-linking modulo rewitnessing. The proof that this is an isomorphism will essentially rely on the ``true'' Yoneda isomorphism.
These results imply that $\exists$-linkings form a $^{*}$-autonomous category in which $\forall X(X\multimap X)$ is the tensor unit and provide a faithful representation of coends.

In the category of $\exists$-linkings the Yoneda isomorphism is a true isomorphism and the ``Russell-Prawitz'' isomorphisms like $A\otimes B\simeq \forall X((A\multimap B\multimap X)\multimap X)$ hold. The representation of initial algebras falls outside the scope of the fragment $\mathsf{MLL}_{\C Y}$, due to the more complex shape of the formulas involved. However, following the ideas in \cite{Uustalu2011}, a generalization of the approach here presented might yield similar results for the representation of initial algebras.

\paragraph{Related work}
Dinaturality is a well-investigated property of System $\mathsf F$ and is usually related to parametric polymorphism (see \cite{Bainbridge1990, Plotkin1993}). The connections between dinaturality, coherence and proof nets are well-investigated in the case of $\mathsf{MLL}$, with or without units (\cite{Blute1991, Blute1993, Blute1996, Lamarche2004, Hughes2012, Heij2016, Mellies2012, Hughes2017}). 
An extensive literature exists on coends in monoidal categories (see \cite{Loregian} for a survey). String diagram representations of some coends can be found in the literature on Hopf algebras and their application to conformal field theory (\cite{Kerler2001, Fuchs2012}). Such coends are all of the restricted form considered in this paper and their representation seems comparable to the one here proposed. 
A different approach to quantifiers as ends/coends over a symmetric monoidal closed category appears in \cite{Mellies2016}, through a bifibrational reformulation of the Lawvere's presheaf hyperdoctrine in the 2-category of distributors. 

The universality problem for the ``Russell-Prawitz'' translation is related to the \emph{instantiation overflow} property (\cite{Ferreira2013}), by which one can transform the System $\mathsf F$ proofs obtained by this translation into proofs in $\mathsf F_{at}$ or \emph{atomic} System $\mathsf F$, which have the desired properties (see \cite{Ferreira2009}). In \cite{PistoneInsta} is shown that the atomized proofs are equivalent to the original ones modulo dinaturality.  $\exists$-linkings provide a very simple approach to instantiation overflow, to be investigated in the future, as the transformation from $\mathsf F$ to $\mathsf F_{at}$ corresponds to rewitnessing.

The representation of proof nets here adopted is inspired from results on $\mathsf{MLL}$ with units (\cite{Trimble,Blute1996,Hughes2012}) and on $\mathsf{MLL1}$ (\cite{Hughes2018}). Proof nets for first-order and second order quantifers were first conceived by means of boxes (\cite{linear}). Later, Girard proposed two distinct boxes-free formalisms (in \cite{quanti1987, quanti1990} for $\mathsf{MLL1}$ but extendable to $\mathsf{MLL2}$, see \cite{tortoraphd}), the second of which is referred here as ``Girard nets''. Different refinements of proof nets for $\mathsf{MLL1}$ and $\mathsf{MLL2}$ have been proposed (\cite{MKinley, Hughes2018} for $\mathsf{MLL1}$ and \cite{Strass2009} for $\mathsf{MLL2}$) to investigate variable dependency issues related to Herbrand theorem and unification, which are not considered here. 

\section{Girard nets and their interpretation in dinatural models}

We let $\C L^{2}$ be the language generated by a countable set of variables $X,Y,Z,\dots \in\TT{Var}$ and their negations $X^{\bot},Y^{\bot},Z^{\bot}, \dots$ and the connectives $\otimes, \parr, \forall, \exists$. Negation is extended in an obvious way into an equivalence relation over formulas. 
By sequents $\Gamma,\Delta, \dots$ we indicate finite multisets of formulas. A sequent $\Gamma$ is \emph{clean} when no variable occurs both free and bound in $\Gamma$ and any variable in $\Gamma$ is bound by at most one $\forall $ or $\exists$ connective. 

By $\mathsf{MLL2}$ we indicate the standard sequent calculus over $\C L^{2}$.
\cite{quanti1990} describes proof nets for first-order $\mathsf{MLL}$. Both the description of proof structures and the correctness criterion can be straightforwardly turned into a definition of proof structures and proof nets for $\mathsf{MLL2}$ (see for instance \cite{tortoraphd}). We indicate the latter as \emph{Girard proof structures} and \emph{Girard nets} (shortly, $G$-proof structures and $G$-nets\footnote{In \cite{quanti1990} the definition of proof structures is based on two conditions: (1) that any $\forall $ link has a distinct eigenvariable and (2) that the conclusions of a proof structures have no free variable (in particular, new constants $\OV x$ are introduced to eliminate free variables). Moreover, in the definition of the correctness criterion any $\forall$-link of eigenvariable $X$ can \emph{jump} on any formula in which $X$ occurs free. 
In \cite{Hughes2018} conditions (1) and (2) are replaced by the equivalent condition that the conclusions of the proof structure plus the witnesses of existential links must form a \emph{clean} sequent and the correctness criterion is modified by demanding that a $\forall$-link of eigenvariable $X$ can \emph{jump} on any $\exists$-link whose witness formula contains free occurrences of $X$. Here we will consider this formulation.\label{jumps}}). 
 We let $\BB G$ indicate the \emph{category of $G$-nets}, whose objects are the types of $\mathsf{MLL2}$ and where $\BB G(A,B)$ is the set of cut-free $G$-nets of conclusions $A^{\bot},B$ (with composition given by cut-elimination).

%
%


Some useful definitions and properties of $^{*}$-autonomous categories and coends can be found in appendix \ref{appB}.
It is well-known (see \cite{Lamarche2006}) that, if we let $\BB P$ be the category of $\mathsf{MLL}$ proof nets and $\BB C$ be any (strict) $^{*}$-autonomous category, then any map $\varphi: \TT{Var}\to Ob_{\BB C}$ generates a functor $\Phi: \BB P \to \BB C$. 
We will now extend this result to $\mathsf{MLL2}$ by considering \emph{dinatural models}, that is, models in which $\mathsf{MLL2}$ proofs are interpreted as {dinatural transformations} \cite{Bainbridge1990}. 
We show how any $G$-net can be interpreted in a dinatural model over a $^{*}$-autonomous category $\BB C$, and we deduce that any map $\varphi: \TT{Var} \to Ob_{\BB C}$ generates a functor $\Phi: \BB G\to \BB C$.


%

It is well-known that dinatural transformations do not compose. The standard approach to interpret second order proofs (see \cite{Bainbridge1990}) is thus to restrict to a class of composable dinatural transformations. In order to interpret quantifiers one considers then \emph{relativized} ends/coends, i.e.\ wedges/co-wedges (see appendix \ref{appB}) which are universal among the class of dinatural transformations in the model.

\begin{definition}[dinatural model]\label{dinamodel}
Let $\BB C$ to be a (strict) $^{*}$-autonomous category $\BB C$. A \emph{dinatural model over $\BB C$} is a category $\C F$ such that 
\begin{itemize}
\item the objects of $\C F$ are multi-variant functors over $\BB C$, including projections of any arity and 
the constant functor $\B 1^{\BB C}$, and closed with respect to 
$\otimes$ and $^{*}$;
\item for all objects $F,G$, $\C F(F,G)$ is a set of dinatural transformations from $F$ to $G$, so that
$\C F$ is $^{*}$-autonomous with unit $ \B 1^{\BB C}$, monoidal product $\otimes$ and involution $^{*}$;
\item the objects of $\C F$ contain all ends and coends relativized to arrows in $\C F$.
\end{itemize}

\end{definition}

The definition above can be recast in the standard fibrational setting of second order models (see \cite{Seely1990}) by using properties of ends and coends.
Two dinatural models are suggested in \cite{Blute1993} and \cite{lauchli}. Moreover, a \emph{free} dinatural model is obtained by quotienting the syntactic model of $\mathsf{MLL2}$ under the congruence generated by all equations expressing the fact that quantifiers correspond to wedges and co-wedges.

In the rest of this section we suppose given a dinatural model $\C F$ over a (strict) $^{*}$-autonomous category $\BB C$ .
Any formula $A\in \C L^{2}$ whose free variables are within $X_{1},\dots, X_{n}$ can be interpreted as a functor $A^{\BB C, \C F}: (\BB C^{op}\times \BB C)^{n}\to \BB C$ in $\C F$ by letting
$$
\begin{matrix}
X_{i}^{\BB C, \C F}(\vec a, \vec b):= b_{i} & X_{i}^{\BB C, \C F}(\vec f, \vec g):= g_{i} \\
(A\otimes B)^{\BB C, \C F}:= A^{\BB C, \C F}\otimes B^{\BB C, \C F} &
(\forall YA)^{\BB C, \C F} := \int_{y}^{\C F}A^{\BB C, \C F}(y,y) &
(A^{\bot})^{\BB C, \C F}:= (A^{\BB C, \C F})^{*}
\end{matrix}
$$
where $\int_{y}^{\C F}F$ indicates the end relativized to $\C F$. In the following lines, since reference to $\C F$ is clear, we will  write $A^{\BB C, \C F}$ as $A^{\BB C}$ and $\int_{y}^{\C F}F$ as $\int_{y}F$ for simplicity.
For a clean sequent $\Gamma=A_{1},\dots, A_{n}$, whose free variables are within $X_{1},\dots, X_{n}$, we let $\Gamma^{\BB C}:= A_{1}^{\BB C}\parr \dots \parr A_{n}^{\BB C}$ (where $x\parr y:= \BB C(x^{\bot},y)$) if $n\geq 1$ and $\Gamma^{\BB C}=\B 1_{\BB C}$ if $n=0$. 

\begin{lemma}[substitution lemma]\label{subst}
$(A[B/X])^{\BB C}(x,x)= A^{\BB C}(B^{\BB C}(x,x), B^{\BB C}(x,x))$.
\end{lemma}
\begin{proof}
Induction on $A$. The only delicate case is $A=\forall YA'$, and, as we can suppose that $B^{\BB C}$ does not depend on $y$, $(A[B/X])^{\BB C}(x,x)= 
\int_{y} ((A'[B/X])^{\BB C}((y,x),(y,x)))\stackrel{[i.h.]}{=}
\int_{y}(A')^{\BB C}((y, B^{\BB C}), (y, B^{\BB C})) =
(\int_{y}(A')^{\BB C}((y,x),(y,x)))(B^{\BB C}, B^{\BB C})=
A^{\BB C}(B^{\BB C}, B^{\BB C})$.

\end{proof}



Let $\pi$ be a cut-free $G$-net of conclusions $\Gamma$ and let all formulas occurring in $\pi$ be within $X_{1},\dots, X_{n}$. We now show that $\pi$ can be interpreted as a dinatural transformation $\pi^{\BB C, \C F}:  \B 1_{\BB C}\to \Gamma^{\BB C, \C F}$\footnote{As explained in appendix \ref{appB}, we omit for readability reference to variables $x_{1},\dots, x_{n}$.}.
As in the case of functors, since reference to $\C F$ is clear, we will simply write $\pi^{\BB C, \C F}$ as $\pi^{\BB C}$.  
Similarly to \cite{Lamarche2006} (Th. 2.3.1. p. 32), we can define $\pi^{\BB C}$ by induction on a sequentialization of $\pi$. We adopt a sequentialization theorem for $G$-nets inspired from \cite{Hughes2018} and described in appendix \ref{AppC}. 
\begin{itemize}
\item if $\pi$ is an axiom link of conclusions $X^{\bot},X$, then $\pi^{\BB C}:= \hat{\B 1}_{A^{\BB C}}$.
%
%
%
%
\item if $\Gamma=\Delta,A\parr B$ and $\pi$ is obtained from a $\pi'$ of conclusions $\Delta, A,B$ by adding a $\parr$-link, then $\pi^{\BB C}:= (\pi')^{\BB C}$.
\item if $\Gamma= \Delta_{1},\Delta_{2},A\otimes B$ and $\pi$ is obtained from $\pi_{1}$ of conclusions $\Delta_{1},A$ and $\pi_{2}$ of conclusions $\Delta_{2}, B$, then $\pi^{\BB C}:= t_{\vec x}\circ \big ((\pi_{1})^{\BB C}\otimes (\pi_{2})^{\BB C}\big )$, where $t_{\vec x}: (\Delta_{1}^{\BB C}\parr A^{\BB C})\otimes (\Delta_{2}^{\BB C}\parr B^{\BB C})\to \Delta_{1}^{\BB C}\parr \Delta_{2}^{\BB C}\parr (A\otimes B)^{\BB C}$ is $\iota_{A^{\BB C}, \Delta_{1}^{\BB C}, (\Delta_{2}\parr B)^{\BB C }}  \circ (\iota_{A^{\BB C}, \Delta_{2}^{\BB C}, B^{\BB C}  } \parr B^{\BB C}  )$, given the natural transformation $\iota_{a,b,c}:(a\parr b)\otimes c\to (a\otimes c)\parr b$.
\item if $\Gamma=\Delta,\forall YA$ and $\pi$ is obtained from $\pi'$ of conclusions $\Delta,A$, then from
$(\pi')^{\BB C}_{x}: \B 1_{\BB C}\to \Delta^{\BB C}\parr A^{\BB C}$ we obtain (by applying the natural isomorphism $\BB C(a\otimes b^{\bot},c)\simeq \BB C(a, b\parr c)$) a dinatural transformation
$\theta_{x}: (\Delta^{\BB C})^{\bot}\to A^{\BB C}$ \footnote{More precisely, $\theta_{x}$ is $\theta_{x_{1},\dots, x_{n},x}$ and comes  from $(\pi')^{\BB C}_{x_{1},\dots, x_{n},x}$, where $(\Delta^{\BB C})^{\bot}$ does not depend on $x$.}. $\pi^{\BB C}$ is now obtained by the universality of (relativized) ends, as shown by the diagram below:

\adjustbox{scale=0.75,center}{
\begin{tikzcd}
    (\Delta^{\BB C})^{\bot} \ar[bend left=15]{rrd}{\theta_{a}} \ar[bend right=15]{ddr}[below]{\theta_{b}} \ar[dashed]{rd}{\pi^{\BB C}}&  &  \\
     &   \int_{y}A^{\BB C}(y,y) \ar{r}{ \delta_{a}^{A^{\BB C}}} \ar{d}{ \delta_{b}^{A^{\BB C}}} &   A^{\BB C}(a,a) \ar{d}{A^{\BB C}(a,f)} \\ 
    &  A^{\BB C}(b,b) \ar{r}[below]{ A^{\BB C}(f,b)} &  A^{\BB C}(a,b)
\end{tikzcd}}

\item if $\Gamma=\Delta,\exists YA$ and $\pi$ is obtained from $\pi'$ of conclusions $\Delta, A[B/X]$, then $\pi^{\BB C}$ is obtained from $(\pi')^{\BB C}$ by the chain of arrows below (by exploiting lemma \ref{subst}):

\adjustbox{scale=0.8, center}{
\begin{tikzcd}
\B 1_{\BB C}  \ar{r}{(\pi')^{\BB C}} &
\Delta^{\BB C}\parr A^{\BB C}(B^{\BB C},B^{\BB C}) \ar{r}{ \omega_{B^{\BB C}}^{\Delta^{\BB C}\parr A^{\BB C}} } &
\int^{x}(\Delta^{\BB C}\parr A^{\BB C}(x,x)) \ar{r}{\nu} & \Delta^{\BB C}\parr \int^{x}A^{\BB C}(x,x)
\end{tikzcd}}
where $\nu$ is given in equation \ref{coends2} in appendix \ref{appB}. 
\end{itemize}

\begin{remark}\label{rem:coend}
It is well-known that $\mathsf{MLL}$ proof nets can be interpreted as (composable) dinatural transformations over any $^{*}$-autonomous category $\BB C$ \cite{Blute1993}, without requiring a dinatural model over $\BB C$ to exist. This fact does not seem to scale to $\mathsf{MLL2}$, since the last step of the definition above exploits the composition of two dinatural transformations.

\end{remark}

We show now that the definition of $\pi^{\BB C}$ does not depend on the sequentialization chosen. 
We must consider all possible permutations of rules in a sequentialization of $\pi^{\BB C}$. We call a $\exists$ link \emph{simple} if it has no incoming jump. For readability we will often confuse formulas $A$ and proof nets $\pi$ with their interpretations $A^{\BB C}$ and $\pi^{\BB C}$.

\begin{itemize}

\item permutations between $\parr,\forall$ and simple $\exists$:
	\begin{description}
	\item[($\parr/\parr$)] We can argue as in \cite{Lamarche2006}.
	\item[($\parr/\forall$)] $\pi_{1},\pi_{2}$, of conclusions $\Gamma, A\parr B, \forall XC$ come from $\pi'$ of conclusions $\Gamma, A,B, C$. The claim follows from the fact that the introduction of $\parr$ does not change the interpretation.
	\item[($\forall/\forall$)] $\pi_{1},\pi_{2}$, of conclusions $\Gamma, \forall XA, \forall YB$ come from $\pi'$ of conclusions $\Gamma, A,B$. The claim follows from $\int_{x}A^{\BB C}(x,x) \parr \int_{y}B^{\BB C}(y,y)\  \stackrel{Eq. \ \ref{commut}}{\simeq}\ 	\int_{x}\int_{y}(A^{\BB C}(x,x)\parr B^{\BB C}(y,y))\  \stackrel{Eq. \ \ref{fubini}}{\simeq}\ 
	\int_{y}\int_{x}(A^{\BB C}(x,x)\parr B^{\BB C}(y,y))\ \stackrel{Eq.\  \ref{commut}}{\simeq}\ \int_{x}A^{\BB C}(x,x) \parr \int_{y}B^{\BB C}(y,y)$.

	\item[($\parr/\exists$)] Similar to case $(\parr/\forall)$.
	\item[($\forall/\exists$)] $\pi_{1},\pi_{2}$ of conclusions $ \forall XA, \exists YB$ (we omit contexts $\Gamma$ for simplicity) come from $\pi'$ of conclusions $ A, B[C/Y]$, where $C$ has no free occurrence of $X$. We let $c=C^{\BB C}$, $\theta$ indicate the translation of the $G$-net of conclusions $  \forall XA, B[C/Y]$ and $\sigma_{x}$ indicate the translation of the $G$-net of conclusions $ A,\exists YB$, so that $\pi_{1}= (\int_{x}A\parr\omega_{c}^{B})\circ \theta$ and $\pi_{2}$ is the universality arrow in the dinaturality diagram for $\sigma_{x}$. Then $\pi_{1}=\pi_{2}$ follows from the universality of $\pi_{2}$, as shown by the diagram below:

	\adjustbox{scale=0.7,center}{	
$	
\begin{tikzcd}	
\B 1_{\BB C}   \ar[bend left=20]{rrrrdd}{\sigma_{a}} \ar[bend left=10]{rrrd}{\pi'_{a}}  \ar[bend right=50]{ddddrr}[near end]{\sigma_{b}}  \ar[bend right=10]{rddd}{\pi'_{b}}  \ar[dashrightarrow, bend right=10]{rd}{\theta}  \ar[dashrightarrow, bend left=10]{rrdd}{\pi_{2}}  &       &   &   &     \\
  &  \int_{x}A\parr B(c,c)  \ar{rr}{\delta^{A}_{a}\parr B(c,c)}   \ar{dd}{\delta^{A}_{b}\parr B(c,c)} \ar{rd}[below]{\int_{x}A\parr \omega^{B}_{c}}  &   &   A(a,a)\parr B(c,c)  \ar{dd}[near end]{A(a,f)\parr B(c,c)}  \ar{rd}{A\parr \omega^{B}_{c}}  &  	 \\
&  &   \int_{x}(A\parr \int^{y}B)\simeq \int_{x}A\parr \int^{y}B   \ar{rr}[near start]{\delta_{a}^{A}\parr \int^{y}B}  \ar{dd}[near end]{\delta_{b}^{A}\parr \int^{y}B}   &   &  A(a,a)\parr \int^{y}B \ar{dd}{A(a,f)\parr \int^{y}B} \\
 & A(b,b)\parr B(c,c)  \ar{rr}[near end]{A(f,b)\parr B(c,c)} \ar{rd}{A\parr \omega_{c}^{B}}  &   &   A(a,b)\parr B(c,c)  \ar{rd}{A\parr \omega_{c}^{B}}&  \\  
  &   &   A(b,b)\parr \int^{y}B  \ar{rr}[below]{A(f,b)\parr \int^{y}B}  &  &   A(a,b)\parr \int^{y}B
\end{tikzcd}	
$}	
	
	\item[($\exists/\exists$)] Similar to case $(\forall/\forall)$.

	\end{description}

\item permutations between a splitting $\otimes$ and $\parr,\forall$ or simple $\exists$:

\begin{description}
	\item[($\otimes/\parr$)] We can argue as in \cite{Lamarche2006}.
	\item[($\otimes/\forall$)] $\pi_{1},\pi_{2}$, of conclusions $ A\otimes C, \forall XB$ (we omit contexts $\Gamma,\Delta$ for simplicity) are obtained from $\sigma$, of conclusions $ A, B$ and $\tau$, of conclusion $ C$, so that 
	$\pi_{1}= \iota_{A,\int_{x}B,C}\circ (\int_{x}\sigma\otimes \tau)$, where $\int_{x}\sigma$ is the interpretation of the $G$-net obtained from $\sigma$ by adding a $\forall$-link and 
	$\pi_{2}$ is the universality arrow in the universality diagram for $\iota_{A, B(x), C}\circ (\sigma_{x}\otimes \tau)$.  Then $\pi_{1}=\pi_{2}$ follows from the universality of $\pi_{2}$, as shown by the diagram below.	
	
\end{description}

\adjustbox{scale=0.7, center}{	
$
\begin{tikzcd}
\B 1_{\BB C} \ar[dashrightarrow, bend left=10]{rrdd}{\pi_{2}}  \ar[bend left=10]{rrrd}{\sigma_{a}\otimes \tau} \ar[bend right=20]{dddr}[below]{\sigma_{b}\otimes \tau} \ar{rd}[below]{\int_{x}\sigma\otimes \tau}&   &   &  &  \\
&  (A\parr \int_{x}B)\otimes C   \ar{rr}{(A\parr \delta_{a}^{B})\otimes C}   \ar{dd}{(A\parr \delta_{b}^{B})\otimes C}  \ar{rd}[below]{\iota_{A,\int_{x}B, C}}  &    &   (A\parr B(a,a))\otimes C   \ar{rd}{\iota_{A,B(a,a),C}}  & \\ 
 &  &   (A\otimes C)\parr \int_{x} B \simeq \int_{x}((A\otimes C)\parr B)    \ar{rr}[below]{(A\otimes C)\parr \delta_{a}^{B}}  \ar{dd}{(A\otimes C)\parr \delta_{b}^{B}}    &   &     (A\otimes C)\parr B(a,a)  \ar{dd}{(A\otimes C)\parr B(a,f)} \\
 &
  (A\parr B(b,b))\otimes C   \ar{rd}[below]{\iota_{A,B(b,b),C}}
   & & &
 \\
&   &  (A\otimes C)\parr B(b,b) \ar{rr}[below]{(A\otimes C)\parr B(f,b)}  & &  (A\otimes C)\parr B(a,b)
\end{tikzcd}
$}

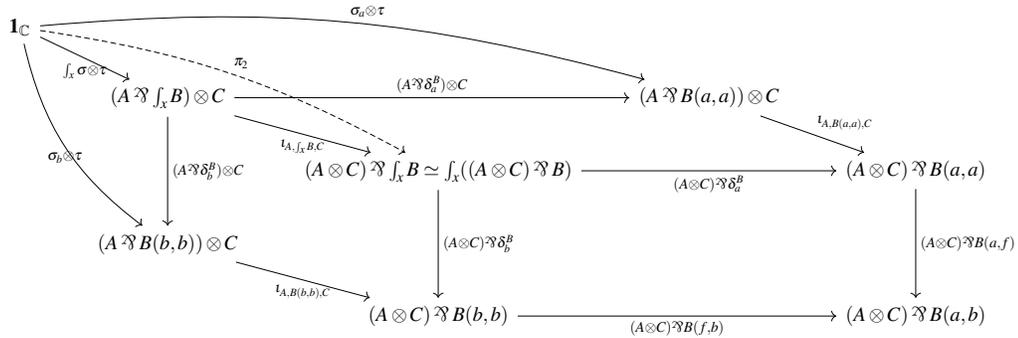
\begin{description}

	\item[($\otimes/\exists$)]
	$\pi_{1},\pi_{2}$, of conclusions $A\otimes D, \exists XB$ (again, we omit contexts $\Gamma, \Delta$ for simplicity) are obtained from $\sigma$, of conclusions $A, B[C/X]$ and $\tau$, of conclusions $D$, so that 
	$\pi_{1}= \iota_{A,\int^{x}B,D}\circ (\int^{x}\sigma\otimes \tau)$, where $c=C^{\BB C}$, $\int^{x}\sigma=(A\parr \omega_{c}^{B})\circ\sigma$ is the interpretation of the $G$-net obtained from $\sigma$ by adding a $\exists$-link and
	$\pi_{2}=( (A\otimes D)\parr \omega_{c}^{B})\circ \iota_{A,B(c,c),D} \circ (\sigma\otimes \tau)$.  Then $\pi_{1}=\pi_{2}$ follows from the naturality of $\iota$, as shown in the diagram below. 

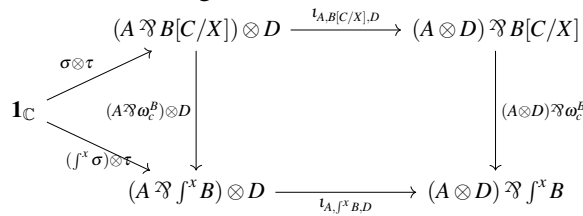
\adjustbox{scale=0.8,center}{	
$
\begin{tikzcd}
& (A\parr B[C/X])\otimes D \ar{dd}[left]{(A\parr \omega_{c}^{B})\otimes D}  \ar{rr}{\iota_{A,B[C/X],D}}  & &  (A\otimes D)\parr B[C/X] \ar{dd}{(A\otimes D)\parr \omega_{c}^{B}} \\
\B 1_{\BB C} \ar{ru}{\sigma\otimes \tau} \ar{rd}[below]{(\int^{x}\sigma)\otimes \tau}&  & & \\
&  (A\parr \int^{x}B)\otimes D \ar{rr}[below]{\iota_{A,\int^{x}B, D}} & &   (A\otimes D)\parr \int^{x}B
\end{tikzcd}
$}	
	
	\end{description}
	
\item permutations between splitting $\otimes$: we can argue as in \cite{Lamarche2006}.
\end{itemize}

The definition above can be extended to the case of a $G$-net with cuts: if $\pi$ has conclusions $\Gamma$ and cut-formulas $B_{1},\dots, B_{n}$, then we can transform $\pi$ into a $G$-net $\pi_{cut}$ of conclusions $\Gamma,[B_{1}\otimes B_{1}^{\bot},\dots, B_{n}\otimes B_{n}^{\bot}]$. Then we can define $\pi^{\BB C}$ as
$(id_{\Gamma^{\BB C}} \parr \hat\bot_{B_{1}^{\BB C}}\parr \dots \parr \hat\bot_{B_{n}^{\BB C}})\circ \pi_{cut}^{\BB C}$. 
The following proposition shows that if the G-net $\pi$ reduces to the cut-free $G$-net $\pi_{0}$, then $\pi^{\BB C}=\pi_{0}^{\BB C}$. Hence it shows that 
the denotation $\pi^{\BB C}$ is invariant with respect to reduction.

\begin{proposition}\label{cuts}
Let $\pi$ be a $G$-net with cuts of conclusions $\Gamma$ and $\pi_{0}$ be the $G$-net obtained from $\pi$ by eliminating all cuts. Then $\pi^{\BB C}=\pi_{0}^{\BB C}$. 
\end{proposition}
\begin{proof}
We consider a reduction sequence of $\pi$ which follows a sequentialization, hence such that any time a cut is eliminated, this cut corresponds to a splitting tensor of $\pi$. As this reduction sequence is finite and terminates on $\pi_{0}$ (by strong normalization and confluence), we can argue by induction on its length. The cases of $\mathsf{MLL}$ cuts can be treated by arguing as in the proof of Lemma 2.3.4, p. 36, of \cite{Lamarche2006}.
We consider then the case of a cut $\forall/\exists$. Let $\pi$ be a $G$-net of conclusions $\Gamma, [\forall XA\otimes \exists XA^{\bot}]$ and let $\pi'$ be the $G$-net of conclusions $\Gamma, [A[B/X]\otimes A^{\bot}[B/X]]$ obtained by applying one reduction step to $\pi'$. We must show that $\sigma_{1}=(\Gamma^{\BB C}\parr \hat\bot_{\int_{x}A^{\BB C}(x,x)})  \circ\pi^{\BB C}$ is equal to $\sigma_{2}= (\Gamma^{\BB C}\parr \hat\bot_{A^{\BB C}(b,b)})\circ (\pi')^{\BB C}$, where $b=B^{\BB C}$. Since the $\otimes$-link is splitting, $\Gamma=\Gamma_{1},\Gamma_{2}$ and $\pi$ (resp. $\pi'$) splits into $\pi_{1}$ of conclusions $\Gamma_{1},\forall XA$ (resp $\pi'_{1}$ of conclusions $\Gamma_{1}, A[B/X]$) and $\pi_{2}$ of conclusions $\Gamma_{2},\exists XA^{\bot}$ (resp. $\pi'_{2}$ of conclusions $\Gamma_{2},A^{\bot}[B/X]$). The claim follows then from the induction hypothesis and the commutation of the diagram below, which is a consequence of the dinaturality of $\hat\bot_{x}$ and of the fact that $\omega_{b}^{A^{\bot}}=(\delta_{b}^{A})^{\bot}$ (as before, for readability we confuse formulas $A$ and proof nets $\pi$ with their interpretations $A^{\BB C}$ and $\pi^{\BB C}$).

\adjustbox{scale=.8, center}{
$
\begin{tikzcd}
&   &    &  A(b,b)\otimes A^{\bot}(b,b)  \ar{rd}{\hat\bot_{A(b,b)}}  &  \\
\Gamma_{1}^{\bot}\otimes \Gamma_{2}^{\bot}  
\ar[bend left=10]{rrru}{(\pi'_{1})_{b}\otimes \pi'_{2}} 
\ar{rr}[below]{\pi_{1}\otimes \pi'_{2}}
\ar[bend right=10]{rrrd}[below]{\pi_{1}\otimes \pi_{2}} &  &  \int_{x}A(x,x)\otimes A^{\bot}(b,b)  \ar{ru}{\delta_{b}^{A}\otimes A^{\bot}} \ar{rd}{\int_{x}A\otimes (\delta_{b}^{A})^{\bot}}   & & \bot_{\BB C} \\
 &  &  &  \int_{x}A(x,x)\otimes \int^{y}A^{\bot}(y,y)  \ar{ru}[below]{\hat\bot_{\int_{x}A}}  & 
\end{tikzcd}
$}

\end{proof}



 Any map $\phi:\TT{Var}\to Ob_{\BB C}$ extends into a map $\varphi: \C L^{2}\to Ob_{\BB C}$ by letting $ (A\otimes B)^{\varphi}=A^{\varphi}\otimes  B^{\varphi}$, $ (\forall XA)^{\varphi}= \int_{x}^{\C F}A^{\varphi}(x,x)$ and $( A^{\bot})^{\varphi}=( A^{\varphi})^{\bot}$.
The following can be verified by induction on formulas:
\begin{lemma}\label{lemmuno}
For each map $\phi:\TT{Var}\to Ob_{\BB C}$ and each sequent $\Gamma$,
$\Gamma^{\BB C}(X_{1}^{\phi},\dots, X_{n}^{\phi})=\Gamma^{\phi}$.
\end{lemma}

By letting $\Phi(\pi):= \pi^{\BB C}(X_{1}^{\phi},\dots, X_{n}^{\phi})$, for $\varphi:\TT{Var}\to Ob_{\BB C}$, we finally get:

\begin{theorem}[functor $\Phi:\BB G\to \BB C$]\label{inter}
For all $\varphi:\TT{Var}\to Ob_{\BB C}$ there exists a functor $ \Phi: \BB G\to \BB C$ such that, for all $A\in \BB L^{2}$, $\Phi(A)= A^{\varphi}$.

\end{theorem}

To account for multiplicative units we must introduce \emph{extended $G$-proof structures}, i.e. $G$-proof structures including two links with no premiss and unique conclusions $\B 1$ and  $\bot$, respectively, and with lax thinning edges (in the sense of \cite{Hughes2012}) connecting any occurrence of $\bot$ with a node. Extended $G$-nets are defined with the usual criterion. Cut-elimination extends straightforwardly to extended $G$-nets. Extended $G$-nets can be sequentialized into the sequent calculus for $\mathsf{MLL2}$ with units. 

The interpretation $\pi^{\BB C}$ extends in a straightforward way to extended $G$-nets. When no quantifier appears in an extended $G$-net $\pi$, then this net corresponds to a \emph{lax linking} in the sense of \cite{Hughes2012}, p.22. We will exploit the result contained in \cite{Hughes2012} that the category $\mathsf{Lax}$ of lax linkings modulo rewiring (see section \ref{linki}) is the free $^{*}$-autonomous category.

We can now define the equivalence relation generated by the interpretation of $G$-nets:
\begin{definition}[equivalence $\simeq_{\varepsilon}$]
We let $\simeq_{\varepsilon}$ be the equivalence relation over $G$-nets given by $\pi \simeq_{\varepsilon}\pi'$ iff
$\pi^{\BB C, \C F}=(\pi')^{\BB C, \C F}$, for any dinatural model $\C F$ over a $^{*}$-autonomous category $\BB C$.
We let $\BB G_{\varepsilon}$ be the category of cut-free $G$-nets considered modulo $\simeq_{\varepsilon}$.

\end{definition}

From proposition \ref{cuts} it follows that $\simeq_{\varepsilon}$ includes $\beta\eta$-equivalence (hence it is a congruence). The following example shows that $\simeq_{\varepsilon}$ strictly extends $\beta\eta$-equivalence. In the next section we will consider a more general example related to the Yoneda isomorphism.

\begin{example}
The category $\BB G$ is not $^{*}$-autonomous (while $\BB G_{\varepsilon}$ is). In particular, $\forall X(X^{\bot}\parr X)$ is not a tensor unit in $\BB G$: by composing any $G$-net  in $ \BB G(Y\otimes \forall X(X^{\bot}\parr X), Y)$ with the unique $G$-net in $\BB G(Y, Y\otimes \forall X(X^{\bot}\parr X))$ one cannot get $id_{Y\otimes \forall X(X^{\bot}\parr X)}$. 

%

%
%
%
%
%
%
%
%
%
%
%
%
%
%
%
%
%
%

\end{example}

\section{The Yoneda translation}

We introduce a way to translate proof nets in (a fragment of) $\mathsf{MLL2}$ into proof nets in $\mathsf{MLL}$ which is related to the \emph{Yoneda isomorphism}.
%
The latter is usually stated as a natural bijection $h:\mathsf{Nat}_{\BB C}({\BB C}(a,x),F(x))\simeq F(a)$, where $F:\BB C\to Set$ and $a\in Ob_{\BB C}$. The maps $h$ and $h^{-1}$ are defined by
\begin{equation}
\begin{matrix}
h(\theta_{x})= \theta_{a}(id_{a}) & \ &  (\theta_{x}\in \mathsf{Nat}_{\BB C}(\BB C(a,x),F(x))) \\
(h^{-1}(z))_{x}(f)=F(f)(z) & \ & (z\in F(a), f\in \BB C(a,x))
\end{matrix}
\end{equation}

In a dinatural model $\C F$, if $F,G$ are covariant functors, 
$\C F(F,G)\simeq \C F(\B 1_{\BB C},\int_{x}^{\C F}F(x)\multimap G(x))$ as a consequence of the universality of (relativized) ends and
the Yoneda isomorphism can be restated as the isomorphism below:
\begin{equation}\label{yodin}
h: \C F\left (\B 1_{\BB C},\int_{x}^{\C F}(F\multimap x)\multimap G(x)\right ) \ \simeq \ \C F(\B 1_{\BB C}, G\circ F)
\end{equation}
This isomorphism can be expressed in the language of $\mathsf{MLL2}$ by equivalences of the form
 $ \forall X((C\multimap X)\multimap D[X]) \simeq D[C/X] $, where $D[X]$ is a formula in which $X$ occurs only positively. This leads to the following definition:
 
\begin{definition}[Yoneda formula]
Given a variable $X\in \TT{Var}$ and a formula $A\in \C L^{2}$, $A$ is \emph{Yoneda in $X$} (resp. \emph{co-Yoneda in $X$}) if $A$ (resp. $A^{\bot}$) is of the form $(\bigotimes_{i}^{n}C_{i}\otimes X^{\bot})\parr D[X]$\footnote{Given a formula $A$ and a finite (possibly empty) sequence of formulas $C_{1},\dots, C_{n}$, we indicate by $\bigotimes_{i}^{n}C_{i}\otimes A$ (resp. $\bigparr_{i}^{n}C_{i}\parr A$) the formula $C_{1}\otimes \dots \otimes C_{n}\otimes A$ (resp. $C_{1}\parr \dots \parr C_{n}\parr A$).}, where $X$ does not occur in any of the $C_{i}$ and $D[X]$ has a unique, positive, occurrence of $X$.
\end{definition} 

We let $\C L^{2}_{\C Y}\subset \C L^{2}$ be the language obtained by restricting $\forall$ quantification (resp. $\exists$ quantification) to Yoneda (resp. co-Yoneda) formulas. In other words $\forall XA\in \C L^{2}_{\C Y}$ (resp. $\exists XA\in \C L^{2}_{\C Y}$) only if $A\in \C L^{2}_{\C Y}$ and $A$ is Yoneda in $X$ (resp. co-Yoneda in $X$).
We indicate by $\mathsf{MLL2}_{\C Y}$ the restriction of $G$-nets to $\C L^{2}_{\C Y}$.

The Yoneda isomorphism induces a translation from $\mathsf{MLL2}_{\C Y}$ formulas into propositional formulas:
the \emph{Yoneda translation} $A_{\C Y}$ of a formula $A\in \C L^{2}_{\C Y}$ is the multiplicative formula obtained by replacing systematically $\forall X((\bigotimes_{i}^{n}C_{i}\otimes X^{\bot})\parr D[X])$ by $D[\bigotimes_{i}^{n}C_{i}\otimes \B 1]$ and 
$\exists X((\bigparr_{i}^{n}C_{i}\parr X)\otimes D[X^{\bot}])$ by $D[\bigparr_{i}^{n}C_{i}\parr \bot]$. 
The formulas $\forall X(X^{\bot}\parr X)$ and $\exists X(X\otimes X^{\bot})$ translate the multiplicative units $\B 1, \bot$. 
We let $\C L_{\B 1, \bot}\subset \C L^{2}_{\C Y}$ be the language obtained by restricting $\forall XA$ to $A=X^{\bot}\parr X$ and $\exists XA$ to $A=X\otimes X^{\bot}$. We let $\mathsf{MLL2}_{\B 1, \bot}$ be the restriction of $G$-nets to $\C L_{\B 1, \bot}$.

Let us fix a dinatural model $\C F$ over a $\BB C$. For any formula $A$ Yoneda in $X$, the isomorphism \ref{yodin} takes the form $h_{A}: (\forall XA)^{\BB C,\C F}\to A_{\C Y}^{\BB C, \C F}$\footnote{It is easily seen that the Yoneda isomorphism can be restated for relativized coends in a dinatural model.}. $h_{A}$ can be represented by means of the extended $G$-nets $Yo_{1}^{A}\in \BB G(\forall XA, A_{\C Y})$ and $Yo_{2}^{A}\in \BB G(A_{\C Y}, \forall XA)$ illustrated in figure \ref{yoneda} (where the blue arrows correspond to lax thinning edges). By inspecting the behavior of these $G$-nets with respect to cut-elimination one easily sees that they correspond to $h_{A}$ in the following sense:

\begin{lemma}[Yoneda isomorphism for $G$-nets]\label{ioneok}
Let $A$ be Yoneda in $X$,
\begin{itemize}
\item[1.] For all $G$-net $\pi$ of conclusion $\forall XA$, 
$(Yo_{1}^{A}\circ\pi)^{\BB C,\C F}= h_{A}(\pi^{\BB C, \C F})$.
\item[2.] For all $G$-net $\pi$ of conclusion $\exists XA^{\bot}$, 
$(Yo_{2}^{A}\circ\pi)^{\BB C, \C F}= h_{A}^{-1}(\pi^{\BB C, \C F})$. 
\end{itemize}
\end{lemma}
%

Let $\BB G^{\C Y}$  (resp. $\BB G_{\varepsilon}^{\C Y}$) be the subcategory of $\BB G$ made of $G$-nets (resp. $G$-nets modulo $\simeq_{\varepsilon}$) in the fragment $\mathsf{MLL2}_{\C Y}$.
By using the extended $G$-nets $Yo_{1}^{A},Yo_{2}^{A}$, the Yoneda  translation can be extended into a functor $\B{Yon}: \BB G^{\C Y}\to\mathsf{Lax}$, where $\mathsf{Lax}$ is the category of {lax linkings} for $\mathsf{MLL}$ recalled in the previous section. The functor $\B{Yon}$ associates to a $\C L^{2}_{\C Y}$ formula $A$ its translation $A_{\C Y}$ and to a $G$-net $\pi$ of conclusions $\Gamma$ the lax linking $\B{Yon}(\pi)$ of conclusions $\Gamma_{\C Y}$ obtained by cutting any occurrence of $\forall XA$ (resp. $\exists XA^{\bot}$) in $\pi$ with $Yo_{1}^{A}$ (resp. with $Yo_{2}^{A}$).

More precisely, $\pi_{\C Y}$ is constructed as follows: since $\pi$ is sequentializable, for any $\exists$-link of conclusion $\exists XA$, there exists a sub-net $\pi_{A}$ of conclusions $\Gamma, A[B/X]$ from which $\pi$ can be obtained by first adding the $\exists$-link and then adding other links. Starting from the topmost $\exists$-links in the sequentialization of $\pi$, let us replace the associated sub-nets $\pi_{A}$ with the sub-net $\pi^{*}_{A}$ obtained by  cutting $\pi_{A}$ with $Yo_{A}^{1}$ and then reducing this cut. 
After eliminating all $\exists$-links, the same construction, with $Yo_{A}^{2}$ in place of $Yo_{A}^{1}$ allows to eliminate $\forall$-links. $\pi_{\C Y}$ is clearly independent from the sequentialization chosen. However, by reasoning by induction on the sequentialization order one can be convinced that all cuts so introduced can be eliminated. A simple verification also shows that the transformation just defined is functorial (i.e. it preserves identity and composition).

As a functor from $\BB G^{\C Y}$ to $\mathsf{Lax}$, $\B{Yon}$ is not faithful: for instance, the composition $Yo_{1}^{A}\circ Yo_{2}^{A}$ is not equal to the identity on $\forall XA$, while its translation yields the identity on $A_{\C Y}$.
This implies that the $G$-net representation of the Yoneda isomorphism is not an isomorphism in $\BB G^{\C Y}$.
This is another way to say that the equivalence $\simeq_{\varepsilon}$ strictly extends $\beta\eta$-equivalence of $G$-nets.

However, the Yoneda isomorphism becomes an isomorphism of $G$-nets as soon as we consider these modulo $\simeq_{\varepsilon}$. More generally, by applying the ``true'' Yoneda isomorphism as well as lemma \ref{ioneok}, we obtain the following:

\begin{lemma}\label{faith}
$\B{Yon}$ is faithful as a functor from $\BB G_{\varepsilon}^{\C Y}$ to $\mathsf{Lax}$.
\end{lemma}

In the next section we will introduce a compact representation of $G$-nets which allows to compute the equivalence $\simeq_{\varepsilon}$ in a syntactic way.

%
%
%
%
%
%
%
%
%
%
%
%
%
%

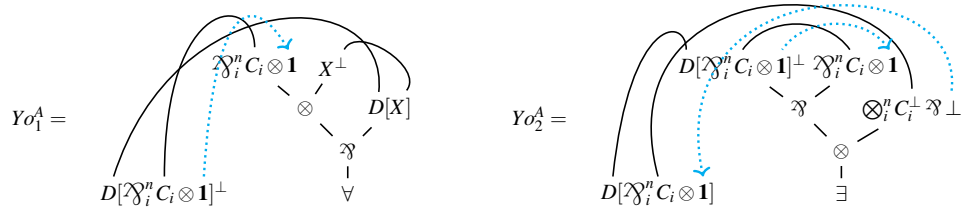
\begin{figure}
%
%
%
\adjustbox{scale=0.75,center}{
$Yo_{1}^{A}= \quad
\begin{tikzpicture}[baseline=-6ex]
\node(d1) at (0,-2.25) {$D[\bigparr_{i}^{n}C_{i}\otimes \B 1]^{\bot}$};

\node(d2) at (1.6,0) {$\bigparr_{i}^{n}C_{i}\otimes \B 1$};
\node(x1) at (3,0) {$X^{\bot}$};
\node(x2) at (4,-0.75) {$D[X]$};

\node(t1) at (2.5,-0.75) {$\otimes$};
\node(p1) at (3.25,-1.5) {$\parr$};
\node(e1) at (3.25,-2.25) {$\forall $};

\draw[thick] (d2) to (t1);
\draw[thick] (x1) to (t1);
\draw[thick] (x2) to (p1);
\draw[thick] (p1) to (t1);
\draw[thick] (p1) to (e1);
\draw[thick] (0,-1.95) .. controls (0,1.3) and  (1.6,1.3) .. (d2);
\draw[thick] (x1) to [bend left=95] (x2);
\draw[thick] (-0.9,-1.95) .. controls (0.2,1.5) and (4,1.5) .. (3.8,-0.45);
\draw[dotted, cyan, very thick, ->] (0.7,-1.95) .. controls (0.9,1.3) and  (1.6,1.3) .. (2.2,0.3);

\end{tikzpicture}
\qquad\qquad
Yo_{2}^{A}= \quad
\begin{tikzpicture}[baseline=-6ex]
\node(d1) at (0,-2.25) {$D[\bigparr_{i}^{n}C_{i}\otimes \B 1]$};

\node(d2) at (1.5,0) {$D[\bigparr_{i}^{n}C_{i}\otimes \B 1]^{\bot}$};
\node(x1) at (3.5,0) {$\bigparr_{i}^{n}C_{i}\otimes \B 1$};
\node(x2) at (4.5,-0.75) {$\bigotimes_{i}^{n}C_{i}^{\bot}\parr\bot$};

\node(t1) at (2.5,-0.75) {$\parr$};
\node(p1) at (3.25,-1.5) {$\otimes$};
\node(e1) at (3.25,-2.25) {$\exists $};

\draw[thick] (d2) to (t1);
\draw[thick] (x1) to (t1);
\draw[thick] (x2) to (p1);
\draw[thick] (p1) to (t1);
\draw[thick] (p1) to (e1);
\draw[thick] (0,-1.95) .. controls (-1,1.8) and (4.5,1.8) .. (x2);
\draw[thick] (3.4,0.3) to [bend right=55] (1.5,0.3);
\draw[very thick, cyan, dotted, <-] (4.1,0.3) to [bend right=55] (2.2,0.3);

\draw[thick] (-0.8,-1.95) .. controls (-0.5,0.9) and (0.5,0.9) .. (0.5,0.3);
\draw[very thick, cyan, dotted, <-] (0.8,-1.95) .. controls (0,1.8) and (5.5,1.8) .. (5.2,-0.45);

\end{tikzpicture}
$}
\caption{$G$-nets for the Yoneda isomorphism}
\label{yoneda}
\end{figure}

\section{Linkings for $\mathsf{MLL2}_{\C Y}$}

In this section we introduce a compact representation of proof nets for $\mathsf{MLL2}_{\C Y}$. We adopt a notion of \emph{linking} inspired from \cite{Hughes2012,Hughes2018} and a notion of \emph{rewiring} inspired from \cite{Blute1996, Heij2014,Hughes2012} (in which the role of thinning edges is given by \emph{witness edges}). In particular, the restriction to $\C L^{2}_{\B 1, \bot}$ yields a formalism which is equivalent to lax linkings for $\mathsf{MLL}$ (lemma \ref{star}).

Given a formula $A$ (resp. a sequent $\Gamma$) we let $tA=(nA,eA)$ (resp. $t\Gamma=(n\Gamma, e\Gamma)$) be its parse tree (resp. parse forest). We will often confuse the nodes of $\Gamma$ with the associated formulas.
Let $\Gamma$ be a clean sequent. An \emph{edge} $e$ is a pair of leaves of $t\Gamma$ consisting in two occurrences of  opposite polarity of the same variables.  
Any $\exists$-link in $t\Gamma$ has a distinguished eigenvariable. A variable is an \emph{existential variable} if it occurs quantified existentially. We will indicate existential variables as $\B X, \B Y, \dots$, to stress that these variables are treated as ``unknown variables''. A formula containing no free occurrences of existential variables will be called a \emph{ground formula}.
Since in all formulas of the form $\exists XA$, $A$ is co-Yoneda in $X$, existential variables come in pairs, called \emph{co-edges}. We let $\Gamma^{\exists}$ be the set of co-edges of $\Gamma$. Any co-edge $c$ is uniquely associated with an existential formula $A_{c}$. For any formula $B$ and co-edge $c$, we say that $B$ \emph{depends on $c$} when $c=(\B X,\B X^{\bot})$ and $\B X$ occurs free in $B$.

A \emph{linking} of $\Gamma$ is a set of disjoint edges whose union contains all but the existential variables of $\Gamma$. A \emph{witnessing function} over $\Gamma$ is an injective function $W:\Gamma^{\exists}\to n\Gamma$, associating any co-edge with a node of $\Gamma$. We will represent witnessing functions by using colored and dotted arrows, called \emph{witness edges}, going from the two nodes of a co-edge $c$ to the formula $W(c)$.  An \emph{$\exists$-linking} over $\Gamma$ is a pair $\ell=(E,W)$, where $E$ is a linking over $\Gamma$ and $W$ is a witnessing function over $\Gamma$. Examples of $\exists$-linkings are shown in fig. \ref{equi}.

Given a witnessing function $W$, we let the \emph{dependency graph of $W$} be the directed graph $D_{W}$ with nodes the co-edges and arrows $c\to c'$ when $W(c)$ depends on $c'$. We call a witnessing function $W$ \emph{acyclic} when the graph $D_{W}$ is directed acyclic. We call $\ell=(E,W)$ \emph{acyclic} when $W$ is acyclic. When $D_{W}$ is acyclic, the witnessing function $W$ allows to associate a ground formula (called a \emph{ground witness}) $GW(c)$ to any co-edge: if $c$ is a leaf of $D_{W}$, then $W(c)$ is a already ground formula, so $GW(c):=W(c)$; otherwise, if $D_{W}$ contains the edges $(c,c_{1}),\dots, (c,c_{n})$, $W(c)$ depends on the existential variables $\B X_{1},\dots, \B X_{n}$ associated to the co-edges $c_{1},\dots, c_{n}$, respectively, then by induction on the well-founded order induced by $D_{W}$, we can suppose the $GW(c_{i})$ well-defined and put
$GW(c):=W(c)[GW(c_{1})/\B X_{1},\dots, GW(c_{n})/\B X_{n}]$.

Acyclic $\exists$-linkings provide a compact representation of $G$-proof structures, since to an $\exists$-linking $\ell=(E,W)$ can be associated a unique $G$-proof structure $\pi(\ell)$ as follows: 
starting from co-edges which are leaves in $D_{W}$, we repeatedly apply to the graph $E\cup t\Gamma$, recursively on $D_{W}$, the \emph{co-edge expansion} operation shown in fig. \ref{recovery}, which instantiates the unknown variable of a co-edge $c$ with its ground witness $GW(c)$. 
 An $\exists$-linking $\ell$ is \emph{correct} when it is acyclic and $\pi(\ell)$ is a $G$-net.

\begin{figure}\begin{subfigure}{0.48\textwidth}
\adjustbox{scale=0.7, center}{
\begin{tikzpicture}[baseline=-6ex]
\node(a) at (0,0) {$\B X$};
\node(b) at (1,0) {$D[\B X]^{\bot}$};
\node(d) at (-1,0) {$\bigparr_{i}C_{i}^{\bot}$};

\node(p1) at (-0.5,-0.75) {$\parr$};
\node(t1) at (0.25,-1.5) {$\otimes$};
\node(e) at (0.25,-2.25) {$\exists$};

\draw[thick] (d) to (p1);
\draw[thick] (a) to (p1);
\draw[thick] (p1) to (t1);
\draw[thick] (b) to (t1);
\draw[thick] (e) to (t1);
\node (cc) at (0.5,0.8) {$c$};

%

\node(aa) at (2,1) {$GW(c)$};
\draw[thick] (1.8, 1.5) to [bend left=15] (aa);
\draw[thick] (aa) to [bend right=15] (2.2,0.5);


%
%

\draw[very thick, dotted, violet] (a) to [bend left=55] (1.2,0.3);
\draw[very thick, dotted, violet, ->, rounded corners=6pt] (0.5,0.5) to [bend left=35] (aa);

\end{tikzpicture}
$\qquad\leadsto\qquad$
\begin{tikzpicture}[baseline=-6ex]
\node(a) at (0,0) {$GW(c)$};
\node(b) at (1.6,0) {$D[GW(c)]^{\bot}$};
\node(d) at (-1.3,0) {$\bigparr_{i}C_{i}^{\bot}$};

\node(p1) at (-0.5,-0.75) {$\parr$};
\node(t1) at (0.25,-1.5) {$\otimes$};
\node(e) at (0.25,-2.25) {$\exists$};

\draw[thick] (d) to (p1);
\draw[thick] (a) to (p1);
\draw[thick] (p1) to (t1);
\draw[thick] (b) to (t1);
\draw[thick] (e) to (t1);

%

\node(aa) at (2.7,1) {$GW(c)$};

\draw[thick] (0.1, 0.5) to [bend right=15] (a); 
\draw[thick] (aa) to [bend right=15] (3,0.5);


\draw[thick] (1,0.3) to [bend left=75] (aa);

\end{tikzpicture}}
\caption{Expansion of a maximal co-edge}
\label{recovery}
\end{subfigure}
\begin{subfigure}{0.5\textwidth}
\adjustbox{scale=0.6, center}{
$\begin{matrix}
\begin{matrix}
\begin{tikzpicture}[baseline=-2ex]
\node(a) at (0,0) {$A$};
\node(b) at (0.5,0.5) {$B$};

\draw[very thick, violet, dotted] (-1.8,-0.5) to [bend left=45] (-1.2,-0.5);
\draw[very thick, violet, dotted, ->] (-1.5,-0.35) to [bend left=55] (a);

\end{tikzpicture}
 & \quad\leadsto\quad &
\begin{tikzpicture}[baseline=-2ex]

\node(a) at (0,0) {$A$};
\node(b) at (0.5,0.5) {$B$};

\draw[very thick, violet, dotted] (-1.8,-0.5) to [bend left=45] (-1.2,-0.5);
\draw[very thick, violet, dotted, ->] (-1.5,-0.35) to [bend left=45] (b);

\end{tikzpicture}
\end{matrix} \\
\ \\
\text{(1) Moving one witness edge (where $W^{-1}(B)=\emptyset$)} \\
\ \\
\begin{matrix}
 \begin{tikzpicture}[baseline=1ex]

\draw[very thick, orange, dotted] (-1.8,-0.5) to [bend left=45] (-1.2,-0.5);
\draw[very thick, orange, dotted, ->] (-1.5,-0.35) to [bend left=55] (0,0);

\draw[very thick, violet, dotted] (0,0) to [bend left=45] (0.6,0);
\draw[very thick, violet, dotted, ->] (0.3,0.15) to [bend left=55] (1.5,0.75);

\end{tikzpicture}
& \quad\leadsto\quad &
\begin{tikzpicture}[baseline=1ex]

\draw[very thick, violet, dotted] (-1.8,-0.5) to [bend left=45] (-1.2,-0.5);
\draw[very thick, violet, dotted, ->] (-1.5,-0.35) to [bend left=55] (0,0);

\draw[very thick, orange, dotted] (0,0) to [bend left=45] (0.6,0);
\draw[very thick, orange, dotted, ->] (0.3,0.15) to [bend left=55] (1.5,0.75);

\end{tikzpicture} 
\end{matrix} \\
\ \\
\text{(2) Swapping two witness edges}
\end{matrix}$
}
\caption{rewitnessing moves}
\label{rewit}
\end{subfigure} \\

\

\begin{subfigure}{0.35\textwidth}
\adjustbox{scale=0.75,center}{
\begin{tikzpicture}
\node(a) at (0,0) {$\exists X((Y^{\bot}\parr X)\otimes X^{\bot})$};
\node(b) at (4,0) {$\forall X((Y\otimes X^{\bot})\parr X)$};

\draw[thick] (-0.7,0.3) to [bend left=35] (3.4,0.3);
\draw[thick] (4.1,0.3) to [bend left=45] (5.3,0.3);
\draw[very thick, dotted, cyan] (0.1,0.3) to [bend left=35] (1.3,0.3);
\draw[very thick, dotted, cyan, ->] (0.7,0.5) to [bend left=45] (4.1,0.3);

\end{tikzpicture}}
\adjustbox{scale=0.65,center}{
\begin{tikzpicture}
\node(a) at (0,0) {$\exists X((Y^{\bot}\parr X)\otimes X^{\bot})$};
\node(b) at (4,0) {$\forall X((Y\otimes X^{\bot})\parr X)$};

\draw[thick] (-0.7,0.3) to [bend left=35] (3.4,0.3);
\draw[thick] (4.1,0.3) to [bend left=45] (5.3,0.3);
\draw[very thick, dotted, orange] (0.1,0.3) to [bend left=35] (1.3,0.3);
\draw[very thick, dotted, orange, ->] (0.7,0.5) to [bend left=45] (3.4,0.3);

\end{tikzpicture}}
\caption{$\sim$-equivalent $\exists$-linkings}
\label{equi}
\end{subfigure}
\ \ \ \ \ \ \ \  \ \ \ \ 
\begin{subfigure}{0.58\textwidth}
\adjustbox{scale=0.65, center}{
\begin{tikzpicture}
\node(a) at (0,0) {$Y^{\bot}$};
\node(b) at (1,0) {$X$};
\node(c) at (2,-0.75) {$X^{\bot}$};

\node(p1) at (0.5,-0.75) {$\parr$};
\node(t1) at (1.25,-1.5) {$\otimes$};
\node(e) at (1.25,-2.25) {$\exists$};

\draw[thick] (a) to (p1);
\draw[thick] (b) to (p1);
\draw[thick] (c) to (t1);
\draw[thick] (t1) to (p1);
\draw[thick] (t1) to (e);

\node(a') at (4,0) {$Y$};
\node(b') at (5,0) {$X^{\bot}$};
\node(c') at (6,-0.75) {$X$};

\node(t2) at (4.5,-0.75) {$\otimes$};
\node(p2) at (5.25,-1.5) {$\parr$};
\node(f) at (5.25,-2.25) {$\forall$};

\draw[thick] (a') to (t2);
\draw[thick] (b') to (t2);
\draw[thick] (c') to (p2);
\draw[thick] (p2) to (t2);
\draw[thick] (p2) to (f);

\draw[thick] (a) to [bend left=55] (a');
\draw[thick] (b) to [bend left=55] (b');
\draw[thick] (c) to [bend left=75] (c');

\end{tikzpicture}
\qquad
\begin{tikzpicture}
\node(a) at (0,0) {$Y^{\bot}$};
\node(b) at (1,0) {$Y$};
\node(c) at (2,-0.75) {$Y^{\bot}$};

\node(p1) at (0.5,-0.75) {$\parr$};
\node(t1) at (1.25,-1.5) {$\otimes$};
\node(e) at (1.25,-2.25) {$\exists$};

\draw[thick] (a) to (p1);
\draw[thick] (b) to (p1);
\draw[thick] (c) to (t1);
\draw[thick] (t1) to (p1);
\draw[thick] (t1) to (e);

\node(a') at (4,0) {$Y$};
\node(b') at (5,0) {$X^{\bot}$};
\node(c') at (6,-0.75) {$X$};

\node(t2) at (4.5,-0.75) {$\otimes$};
\node(p2) at (5.25,-1.5) {$\parr$};
\node(f) at (5.25,-2.25) {$\forall$};

\draw[thick] (a') to (t2);
\draw[thick] (b') to (t2);
\draw[thick] (c') to (p2);
\draw[thick] (p2) to (t2);
\draw[thick] (p2) to (f);

\draw[thick] (a) to [bend left=65] (b);
\draw[thick] (c) to [bend left=75] (a');
\draw[thick] (b') to [bend left=75] (c');

\end{tikzpicture}}
\caption{$\simeq_{\varepsilon}$-equivalent $G$-nets}
\label{equinet}
\end{subfigure}
\caption{$\exists$-linkings and rewitnessing.}
\end{figure}
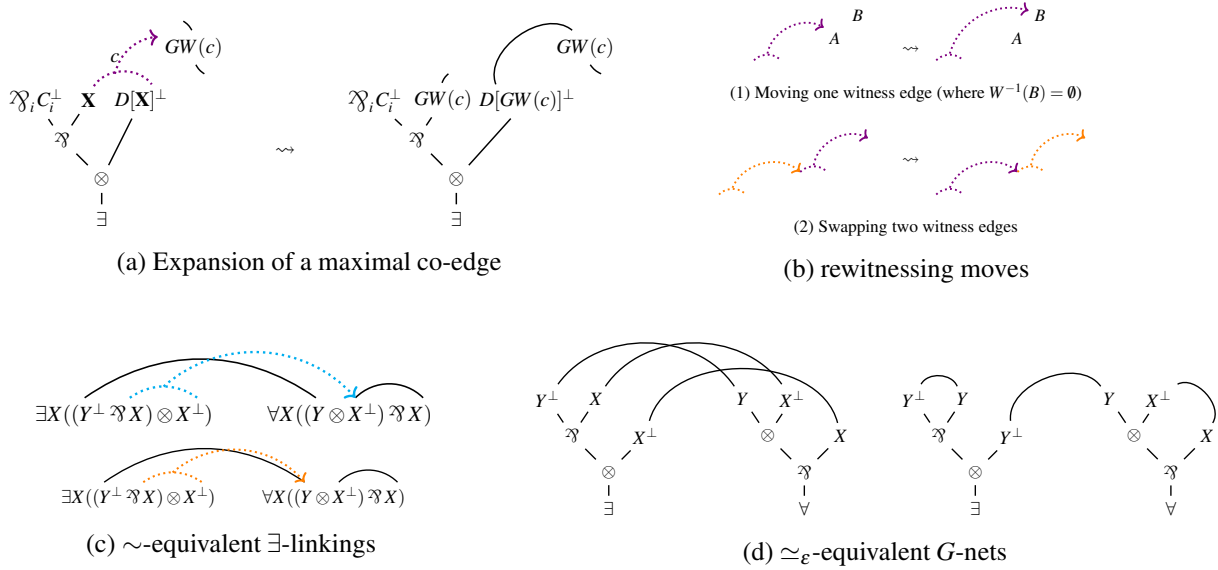

We introduce an equivalence relation over correct $\exists$-linkings, called \emph{rewitnessing}, inspired from the ``rewiring'' technique in \cite{Blute1996, Heij2014, Hughes2012}. Given a witnessing function $W$, a \emph{simple rewitnessing of $W$} is a witnessing function $W'$ obtained by either moving exactly one witness edge from one formula to another ``free'' one (i.e. to some formula $A$ such that $W^{-1}(A)=\emptyset$), or by switching two consecutive witness edges, i.e. two edges $c_{1},c_{2}$ such that $W(c_{1})\in c_{2}$, as shown in fig. \ref{rewit}. We let $\ell \sim_{1} \ell'$ if $\ell=(E,W)$, $\ell'=(E,W')$ and $W'$ is a simple rewitnessing of $W$. We let $\sim$ be the reflexive and transitive closure of $\sim_{1}$. 

In fig. \ref{equi} are shown $\sim$-equivalent $\exists$-linkings over $\exists X((Y^{\bot}\parr X)\otimes X^{\bot}), \forall X((Y\otimes X^{\bot})\parr X)$. These correspond to the two $\simeq_{\varepsilon}$-equivalent $G$-nets in fig. \ref{equinet}. In the next section we will show that rewitnessing can be used to compute the $\varepsilon$-equivalence.
When $A$ is Yoneda in $X$, we let $ID_{\forall XA}$ denote the $\exists$-linking in figure \ref{ide}. 

%


We let $\BB L^{\exists}$ be the \emph{category of $\exists$-linkings}, whose objects are the formulas of $\mathsf{MLL2}_{\C Y}$ and where $\BB L^{\exists}(A,B)$ is the set of $\sim$-equivalence classes of correct $\exists$-linkings of conclusions $A^{\bot},B$, with composition given by cut-elimination (see next section). We let $\BB L^{\B 1,\bot}$ be the restriction of $\BB L^{\exists}$ to $\mathsf{MLL2}_{\B 1, \bot}$ formulas.

Similarly to the functor $\B{Yon}:\BB G\to \mathsf{Lax}$, we can construct a functor $\C Y: \BB L^{\exists}\to Lax$ for $\exists$-linkings. The linking $\ell_{\C Y}$ is obtained in two steps: first, for any co-edge $c=(\B X, \B X^{\bot})$, replace $A_{c}$ by $(A_{c})_{\C Y}$, replace the thinning edge from $c$ to $W(c)$ by a lax thinning edge from $\bot$ to $W(c)$, and move all lax thinning edges pointing to $\B X$ or $\B X^{\bot}$ (or to $\B X\otimes \B X^{\bot}$ if $A_{c}=\bot^{\exists}$) onto $W(c)$; once all co-edges have been eliminated, replace any universal formula $\forall XA$ by $(\forall XA)_{\C Y}$ and eliminate the unique edge $(X^{\bot},X)$. The transformation just described yields then a lax linking $E_{\C Y}$ over the $\mathsf{MLL}$ sequent $\Gamma_{\C Y}$.
Observe that witness edges are replaced by lax thinning edges, see fig. \ref{compare}.


\begin{figure}
\adjustbox{scale=0.75, center}{
$
\ell=\quad
\begin{tikzpicture}[baseline=-6ex]

\node(a) at (0,0) {$A$};
\node(b) at (-2, -1.5) {$\exists$};
\node(t) at (-2,-0.75) {$\otimes$};

\node(c1) at (-2.5,0) {$\B X$};
\node(c2) at (-1.5,0) {$\B X^{\bot}$};

\draw[thick] (t) to (b);
\draw[thick] (t) to (c1);
\draw[thick] (t) to (c2);

\draw[thick] (0.3,0.7) to [bend right=15] (a);
\draw[thick] (0.3,-0.7) to [bend left=14] (a);

\draw[very thick, cyan, dotted] (c1) to [bend left=64] (c2);
\draw[very thick, cyan, dotted, ->] (-2,0.5) to [bend left=44] (a);

\end{tikzpicture}
\qquad\qquad\pi(\ell)=\quad
\begin{tikzpicture}[baseline=-6ex]

\node(a) at (0,0) {$A$};
\node(b) at (-2, -1.5) {$\exists$};
\node(t) at (-2,-0.75) {$\otimes$};

\node(c1) at (-2.5,0) {$A$};
\node(c2) at (-1.5,0) {$A^{\bot}$};

\draw[thick] (t) to (b);
\draw[thick] (t) to (c1);
\draw[thick] (t) to (c2);

\draw[thick] (0.3,0.7) to [bend right=15] (c1);
\draw[thick] (0.3,-0.7) to [bend left=14] (a);

\draw[thick] (c2) to [bend left=35] (a);

\end{tikzpicture}
\qquad\qquad
\ell_{\C Y}=\quad
\begin{tikzpicture}[baseline=-2ex]
\node(a) at (0,0) {$A$};

\node(b) at (-2, 0) {$\bot$};

\draw[very thick, cyan, dotted, ->] (b) to [bend left=44] (a);
\draw[thick] (0.3,0.7) to [bend right=15] (a);
\draw[thick] (0.3,-0.7) to [bend left=14] (a);
\draw[thick] (b) to [bend left=14] (-2.3,-0.7);

\end{tikzpicture}
$
}
\caption{Local comparison of $\ell$, $\pi(\ell)$ and $\ell_{\C Y}$ for $\bot^{\exists}=\exists X(X\otimes X^{\bot})$.}
\label{compare}
\end{figure}
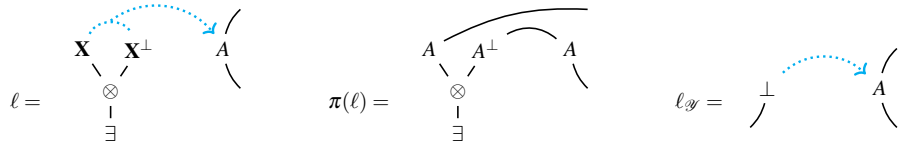

%
%
%
%
%
%

By letting $\sim_{lax}$ denote the rewitnessing equivalence over lax linkings, we have:

\begin{lemma}\label{commuyoneda}
$\ell \sim \ell' \To\ell_{\C Y}\sim_{lax}\ell'_{\C Y}$. 
\end{lemma}
\begin{proof}
The claim follows from the fact that a rewitnessing move of type (1) (fig. \ref{rewit}) in $\ell$ corresponds to a rewiring move in $\ell_{\C Y}$, while a rewitnessing move of type (2) in $\ell$ does not affect $\ell_{\C Y}$. 

\end{proof}

\section{Cut-elimination for $\exists$-linkings}\label{cutelim}

We let a \emph{cut sequent} be a sequent of the form $\Gamma,[\Delta]$, where $\Gamma,\Delta$ is a clean sequent and $\Delta$ is a multiset of formulas, called \emph{cut formulas}, of the form $A\otimes A^{\bot}$ (that we depict by a configuration of the form
 \begin{tikzpicture}[baseline=-1ex]
\node(a) at (0,0) {$A$};
\node(b) at (1.2,0) {$A^{\bot}$};
\draw[thick, rounded corners=4pt] (a) to (0,-0.4) to (1.2,-0.4) to (b);
\end{tikzpicture}). 

By an $\exists$-linking over $\Gamma, [\Delta]$ we indicate an $\exists$-linking over $\Gamma,\Delta$. We call an $\exists$-linking $\ell=(E,W)$ \emph{ready} when $W^{-1}(A)=\emptyset$ for all $A$ occurring in a cut-formula. Cut-elimination relies on the following lemma, proved in appendix \ref{C}.

\begin{lemma}[``ready lemma'']\label{ready}
For any correct $\exists$-linking $\ell$ there exists a ready $\ell'$ such that $\ell'\sim \ell$.
\end{lemma}

Indeed, by lemma \ref{ready} it suffices to apply cut-elimination to ready $\exists$-linkings. \emph{Cut reduction} is the relation over ready $\exists$-linkings defined by the rewrite rules in figure \ref{cut}, where in case \ref{cut2} either $n\geq 1$ or $D[X]\neq X$, and, in case \ref{cut2} and \ref{cut3} the existence of the lefthand edge is forced by the fact that $\Gamma,\Delta$ is clean.
Observe that the reduction $(c)$ incorporates the Yoneda translation. 

\begin{figure}
\begin{subfigure}{0.55\textwidth}
\adjustbox{scale=0.7,center}{
$
\begin{tikzpicture}[baseline=-3ex]
\node(a) at (0,0) {$X$};
\node(b) at (2,0) {$X^{\bot}$};
\node(c) at (4,0) {$X$};

\draw[thick, rounded corners=5pt] (a) to (0,0.5) to (2,0.5) to (b);
\draw[thick, rounded corners=5pt] (b) to (2,-0.5) to (4,-0.5) to (c); 

\end{tikzpicture}
\qquad \leadsto \qquad X
$}
\caption{}
\label{cut11}
\end{subfigure}
\begin{subfigure}{0.49\textwidth}
\adjustbox{scale=0.75,center}{
\begin{tikzpicture}[baseline=-3ex]
\node(b) at (2,0) {$A\otimes B$};
\node(c) at (4,0) {$A^{\bot}\parr B^{\bot}$};

\draw[thick, rounded corners=5pt] (b) to (2,-0.5) to (4,-0.5) to (c); 

\end{tikzpicture}
$\qquad \leadsto \qquad $
\begin{tikzpicture}[baseline=-3ex]
\node(b) at (2,0) {$A\ \ \ \ \  B$};
\node(c) at (4,0) {$A^{\bot}\ \ \ \  B^{\bot}$};

\draw[thick, rounded corners=5pt] (1.7,-0.3) to (1.7,-0.8) to (3.6,-0.8) to (3.6,-0.3); 
\draw[thick, rounded corners=5pt] (2.3,-0.3) to (2.3,-0.6) to (4.3,-0.6) to (4.3,-0.3); 

\end{tikzpicture}
}
\caption{}
\label{cut12}
\end{subfigure}
\begin{subfigure}{\textwidth}
\adjustbox{scale=0.65,center}{
$
\begin{tikzpicture}[baseline=-3ex]
\node(b) at (0,0) {$\forall X( (\bigotimes_{i}^{n}C_{i} \otimes X^{\bot}) \parr D[X])$};
\node(c) at (5,0) {$\exists X((\bigparr_{i}^{n}C^{\bot}_{i}\parr X)\otimes D^{\bot}[X^{\bot}] )$};

\node(d) at (7.5,0.5) {$B$};
\draw[very thick, dotted, orange] (5.3,0.3) to [bend left=45] (6.7,0.3);
\draw[very thick, dotted, orange, ->] (6,0.6) to [bend left=45] (d);

\draw[thick, rounded corners=5pt] (b) to (0,-0.7) to (5,-0.7) to (c); 

\draw[thick, rounded corners=5pt] (0.3,0.3) to (0.3,0.6) to (1.6,0.6) to (1.6,0.3);

\end{tikzpicture}
\qquad \leadsto\qquad
\begin{tikzpicture}[baseline=-2ex]
\node(a) at (-0.5,0) {$D[\bigotimes_{i}^{n}C_{i}\otimes \forall X(X^{\bot}\parr X)]$};
\node(b) at (4.5,0) {$D^{\bot}[\bigparr_{i}^{n}C_{i}^{\bot} \parr \exists X(X\otimes X^{\bot})]$};

\node(c) at (7,0.5) {$B$};

\draw[thick, rounded corners=5pt] (0.2,0.3) to (0.2,0.6) to (1,0.6) to (1,0.3);
\draw[very thick, dotted, orange] (5.3,0.3) to [bend left=45] (6.1,0.3);
\draw[very thick, dotted, orange, ->] (5.7,0.5) to [bend left=35] (c);

\draw[thick, rounded corners=5pt] (a) to (-0.5,-0.6) to (4.5,-0.6) to (b);

\end{tikzpicture}
$}
\caption{}
\label{cut2}
\end{subfigure}
\begin{subfigure}{\textwidth}
\adjustbox{scale=0.65,center}{
$
\begin{tikzpicture}[baseline]
\node(b) at (0,0) {$\forall X( X^{\bot} \parr X)$};
\node(c) at (5,0) {$\exists X( X\otimes X^{\bot} )$};

\node(d) at (7.5,0.5) {$B$};
\draw[very thick, dotted, orange] (4.5,0.3) to [bend left=45] (5.8,0.3);
\draw[very thick, dotted, orange, ->] (5.1,0.6) to [bend left=45] (d);

\draw[thick, rounded corners=5pt] (b) to (-0,-0.7) to (5,-0.7) to (c); 

\draw[thick, rounded corners=5pt] (-0.4,0.3) to (-0.4,0.6) to (0.8,0.6) to (0.8,0.3);

\end{tikzpicture}
\qquad \leadsto\qquad
\begin{tikzpicture}[baseline=2ex]

\node(c) at (3,0.5) {$B$};


\end{tikzpicture}
$}
\caption{}
\label{cut3}
\end{subfigure}
\caption{Cut elimination local steps.}
\label{cut}
\end{figure}
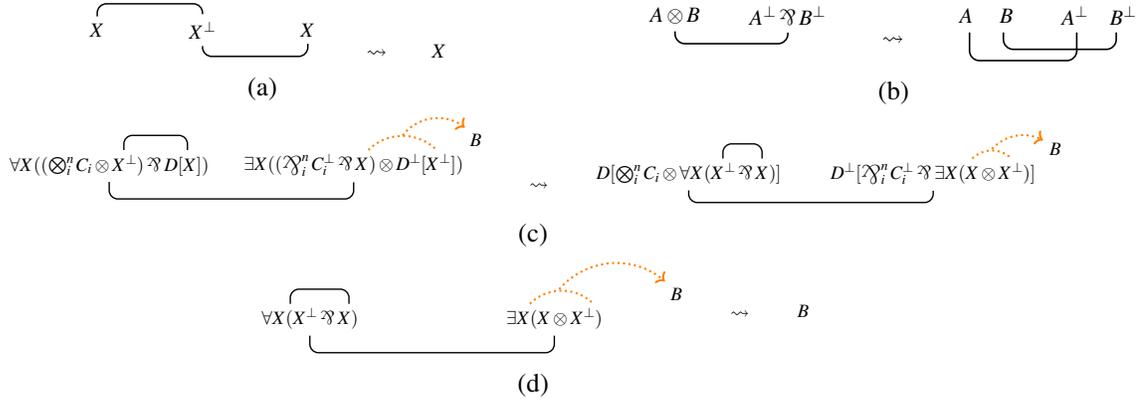

We now verify usual properties of cut-elimination.

\begin{lemma}[confluence]\label{confluence}
Cut reduction is confluent.
\end{lemma}
\begin{proof}
Immediate consequence of the locality of the reduction rules.
\end{proof}

\begin{proposition}[stability]\label{stability}
Let $\ell$ be a correct and ready. If $\ell\leadsto \ell'$, then $\ell'$ is correct.
\end{proposition}
\begin{proof}

For any $G$-net $\pi$ and for any formula $\forall XA$ (with dual formula $\exists XA^{\bot}$) occurring in a cut, let $\pi^{A}$ be the $G$-net obtained by replacing the formula $\forall X A$ (resp. $\exists XA^{\bot}$ ) by $A_{\C Y}$ (resp. $A^{\bot}_{\C Y}$) by cutting it with $Yo_{1}^{A}$ (resp. $Yo_{2}^{A}$).  In other words, we apply the Yoneda translation locally.
$\pi^{A}$ is still a $G$-net, as $\pi$, $Yo_{1}^{A}$ and $Yo_{2}^{A}$ are all sequentializable, and the cut introduced can be applied just after the rules introducing the quantifier of $\forall XA$ (resp. $\exists XA^{\bot}$).

Now, any cut reduction rule $\ell\mapsto \ell'$ induces a transformation of $G$-nets $\pi(\ell) \mapsto^{*} \pi(\ell')$. We must show then that $\mapsto^{*}$ preserves correctness. This is trivial in cases \ref{cut11}, \ref{cut12} and \ref{cut3}.
In case \ref{cut2}, let the cut-formula be $\forall XA\otimes \exists XA^{\bot}$; then $\pi(\ell)\mapsto \pi^{*}$, where $\pi^{*}$ can be obtained from $\pi^{A}$ (which is a $G$-net as $\pi(\ell)$ is a $G$-net and $G$-net reduction preserves correctness) by performing some $G$-net reduction steps. We conclude then that $\pi^{*}$ is correct, i.e. $\ell'$ is correct.

\end{proof}

%
%
%
%
%
%
%
%
%

Strong normalization can be proved in a direct way, without reducibility candidates techniques.

\begin{proposition}[strong normalization]\label{norma}
Let $\ell$ be a correct and ready $\exists$-linking over $\Gamma, [\Delta]$. Then all cut-reductions of $\ell$ terminate over a unique correct $\exists$-linking $nf(\ell)$ over $\Gamma$, called the \emph{normal form of $\ell$}.

\end{proposition}
\begin{proof}
We define a measure $s(A)$ over formulas as follows: $s(X)=s(X^{\bot})=0$, $s(A\otimes B)=s(A\parr B)= s(A)+s(B)+1$, 
$s(\forall X(X^{\bot}\parr X))=s(\exists X(X\otimes X^{\bot})=1$ and, when either $n\geq 1$ or $D[X]\neq X$, $s(\forall X((\bigotimes_{i}^{n}C_{i}\otimes X^{\bot})\parr D[X]))=s(\exists X((\bigparr_{i}^{n}C_{i}^{\bot}\parr X)\otimes D[X]^{\bot}))= s(D[C])+3$, where $C$ is either $\bigotimes_{i}^{n}C_{i}$ or $\bigparr_{i}^{n}C_{i}^{\bot}$.
By letting $s(\ell)$ be the sum all $s(A)$, where $A$ is a cut-formula, any reduction step makes $s(\ell)$ decrease strictly. 
\end{proof}

By proposition \ref{norma} any correct $\exists$-linking has a unique normal form, up to rewitnessing.


\section{Characterization of $\varepsilon$-equivalence}\label{linki}


\begin{wrapfigure}{r}{0.5\textwidth}
\adjustbox{scale=0.78, center}{
$
\begin{tikzpicture}[baseline=-6ex]
\node(a) at (0,0) {$B$};
\node(b) at (1,0) {$D[B]^{\bot}$};
\node(d) at (-1,0) {$\bigparr_{i}C_{i}^{\bot}$};

\node(p1) at (-0.5,-0.75) {$\parr$};
\node(t1) at (0.25,-1.5) {$\otimes$};
\node(e) at (0.25,-2.25) {$\exists$};

\draw[thick] (d) to (p1);
\draw[thick] (a) to (p1);
\draw[thick] (p1) to (t1);
\draw[thick] (b) to (t1);
\draw[thick] (e) to (t1);

\end{tikzpicture}
\ \mapsto \
\begin{tikzpicture}[baseline=-6ex]
\node(a) at (0,0) {$B$};
\node(b) at (1,0) {$D[B]^{\bot}$};

\node(bb) at (2.4,0) {$B$};
\node(bbb) at (3.8,0) {$B^{\bot}$};

\node(d) at (-1,0) {$\bigparr_{i}C_{i}^{\bot}$};

\node(p1) at (-0.5,-0.75) {$\parr$};
\node(t1) at (0.25,-1.5) {$\otimes$};
\node(e) at (0.25,-2.25) {$\exists$};

\draw[thick] (d) to (p1);
\draw[thick] (a) to (p1);
\draw[thick] (p1) to (t1);
\draw[thick] (b) to (t1);
\draw[thick] (e) to (t1);

\draw[thick, bend left=35] (1.2,0.3) to (bb);
\draw[thick, rounded corners=4pt] (bb) to (2.4, -0.4) to (3.6,-0.4) to (3.6,-0.3);

\end{tikzpicture}
$}
\caption{From $\pi$ to $\pi^{cut}$.}
\label{addcut}
\end{wrapfigure}

We exploit the Yoneda translation to prove that the compact representation of $G$-nets by means of $\exists$-linkings characterizes the equivalence induced by ends and coends. We will indeed show that the translation $\ell\to \pi(\ell)$ yields an isomorphism of categories $ \BB L^{\exists} \simeq \BB G_{\varepsilon}^{\C Y}$.

We start by defining the translation $\ell: \pi\mapsto \ell_{\pi}$ ``adjoint'' to $\pi:\ell\mapsto \pi(\ell)$. First, for a $G$-net $\pi$, let $\pi^{cut}$ be obtained from $\pi$ by introducing a new cut for any $\exists$-link of $\pi$ as follows: if $A_{c}=\exists X((\bigotimes_{i}^{n}C_{i}\parr X)\bigotimes D[X]^{\bot})$ with premiss $(\bigotimes_{i}^{n}C_{i}\parr B)\bigotimes D[B]^{\bot}$, introduce an axiom and a cut over $B$ as illustrated in fig. \ref{addcut}. By inspecting the co-edge expansion in fig. \ref{recovery}, it can be seen that
$\pi^{cut}$ is of the form $\pi(\ell^{cut})$ for a unique $\exists$-linking with cuts $\ell^{cut}$. We let then $\ell_{\pi}$ be the normal form of $\ell^{cut} $. While $\ell= \ell_{\pi(\ell)}$ holds by construction, the converse equation $\pi=\pi(\ell_{\pi})$ does not hold in general (since cut-elimination of $\exists$-linking might require rewitnessings).
However, we will show that the weaker $\pi\simeq_{\varepsilon}\pi(\ell_{\pi})$ holds (theorem \ref{equiv}).


We can use the translations $\pi$ and $\ell$ to relate the Yoneda translations for $G$-nets and $\exists$-linkings as follows:

\noindent
\begin{minipage}{0.4\textwidth}
\begin{proposition}\label{ionez}
$a.$ \ $\B{Yon}\circ \pi \ =  \ \C Y$. \\
$b.$ \ $\C Y\circ \ell \ =  \ \B{Yon}$.
\end{proposition}
\end{minipage}
\begin{minipage}{0.6\textwidth}
\adjustbox{scale=0.85, center}{$\begin{tikzcd}
\BB G_{\varepsilon}^{\C Y}  \ar{dr}[left]{\B{Yon}} \ar[bend right=10]{rr}[below]{\ell} & &  \BB L^{\exists} \ar[bend right=10]{ll}[above]{\pi} \ar{dl}{\C Y} \\
 & Lax & 
 \end{tikzcd}$}
\end{minipage}

\begin{proof}
$a.$ can be verified by inspecting the reduction steps involved in the transformation of $\pi(\ell)$ into a lax linking. For $b.$ we argue as follows: $\pi$ is $\beta$-equivalent to $\pi^{cut}=\pi(\ell^{cut})$, where $\ell^{cut}\sim \ell_{\pi}$. Now, from $a.$ it follows that $\B{Yon}(\pi)=\B{Yon}(\pi^{cut})=\B{Yon}(\pi(\ell^{cut}))\sim_{lax}   \ell^{cut}_{\C Y}$.
From $\ell_{\pi}\sim \ell^{cut}$ we deduce then, by lemma \ref{commuyoneda}, that $(\ell_{\pi})_{\C Y}\sim_{lax} \ell^{cut}_{\C Y}$, hence we conclude $(\ell_{\pi})_{\C Y}\sim_{lax} \B{Yon}(\pi)$.
\end{proof}

From proposition \ref{ionez} we deduce that if $\ell$ is correct, $\ell_{\C Y}$ is correct (since $\ell_{\C Y}=\B{Yon}(\pi(\ell))$). Moreover, we deduce that the functor $\C Y$ is faithful (as $\B{Yon}$ is).

The following proposition allows to state that $\ell$ is indeed a functor $\ell:\BB G_{\varepsilon}^{\C Y}\to \BB L^{\exists}$.

\begin{proposition}\label{ellpipi}
If $\pi\simeq_{\varepsilon}\pi'$, then $\ell_{\pi}\sim \ell_{\pi'}$.
\end{proposition}

Proposition \ref{ellpipi} is deduced from the two lemmas below.

\begin{lemma}\label{star}
$\BB L^{\exists}$ is $^{*}$-autonomous. 
 $\BB L^{\B 1, \bot}$ is the free $^{*}$-autonomous category.
\end{lemma}
\begin{proof}
That $\BB L^{\exists}$, with units $\forall X(X^{\bot}\parr X)$ and $\exists X(X\otimes X^{\bot})$, verifies all coherence conditions of a $^{*}$-autonomous category is a simple verification. The second point follows from the faithfulness of $\C Y$ and the fact that $\mathsf{Lax}$ is the free $^{*}$-autonomous category (\cite{Hughes2012}).
\end{proof}

For any $A=(\bigparr_{i}C_{i}\parr X)\otimes D[X^{\bot}]$ Yoneda in $X$ and any $B\in \C L^{2}_{\C Y}$, let $\Omega_{A}^{B}$ be the correct $\exists$-linking in fig. \ref{omega}. Moreover, for all $B,C\in \C L^{2}_{\C Y}$, we let $A(E,\ell)$ and $A(\ell,F)$ be the  correct $\exists$-linking in fig. \ref{precoend}, corresponding to the covariant and contravariant functorial action of $A$ on $\ell$.
The following lemma states then that the existential quantifier behaves like a co-wedge in $\BB L^{\exists}$.
\begin{lemma}\label{coend}
For all $A$ Yoneda in $X$, $E,F\in \C L^{2}_{\C Y}$ and $\ell\in \BB L^{\exists}(E,F)$, $\Omega_{A}^{E} \circ A(\ell,E) \sim \Omega_{A}^{F}\circ A(F,\ell)$
\end{lemma}
\begin{proof}
Indeed $\Omega_{A}^{E} \circ A(\ell,E)$ and $ \Omega_{A}^{F}\circ A(F,\ell)$ differ by a unique rewitnessing, see fig. \ref{coendyes}. 
\end{proof} 

%
%

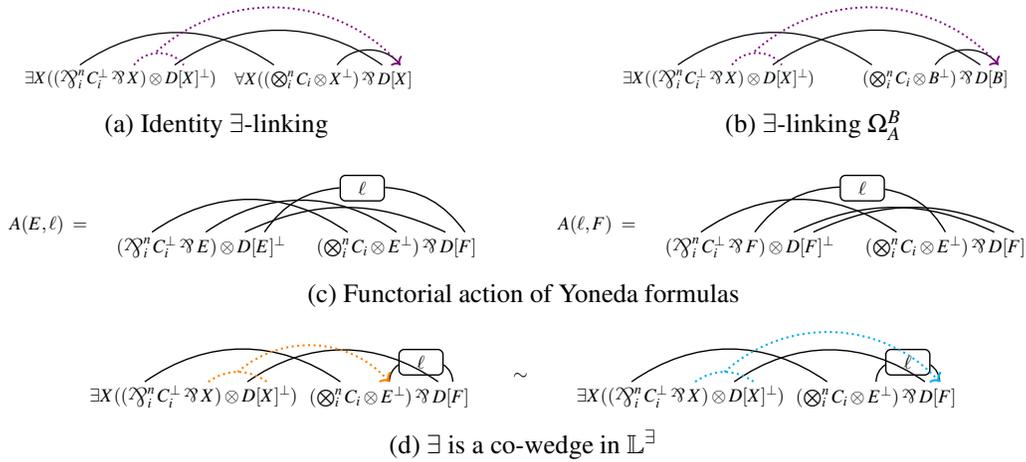
\begin{figure}
\begin{subfigure}{0.49\textwidth}
\adjustbox{scale=0.6, center}{
\begin{tikzpicture}
\node(a) at (0,0) {$\exists X((\bigparr_{i}^{n}C_{i}^{\bot}\parr X)\otimes D[X]^{\bot})$};

\node(b) at (4.5,0) {$\forall X((\bigotimes_{i}^{n}C_{i}\otimes X^{\bot})\parr D[X]$};

\draw[thick] (-0.9,0.3) to [bend left=35] (3.4,0.3);
\draw[thick] (1.2,0.3) to [bend left=35] (5.8,0.3);
\draw[thick] (4.8,0.3) to [bend left=55] (6.2,0.3);
\draw[very thick, dotted, violet] (0.3,0.3) to [bend left=55] (1.4,0.3);
\draw[very thick, dotted, violet, ->] (0.8,0.6) to [bend left=45] (6.2,0.3);

\end{tikzpicture}}
\caption{Identity $\exists$-linking}
\label{ide}
\end{subfigure}
\begin{subfigure}{0.49\textwidth}
\adjustbox{scale=0.6, center}{
\begin{tikzpicture}
\node(a) at (0,0) {$\exists X((\bigparr_{i}^{n}C_{i}^{\bot}\parr X)\otimes D[X]^{\bot})$};

\node(b) at (4.8,0) {$(\bigotimes_{i}^{n}C_{i}\otimes B^{\bot})\parr D[B]$};

\draw[thick] (-0.9,0.3) to [bend left=35] (3.4,0.3);
\draw[thick] (1.2,0.3) to [bend left=35] (5.8,0.3);
\draw[thick] (4.8,0.3) to [bend left=55] (6.2,0.3);
\draw[very thick, dotted, violet] (0.3,0.3) to [bend left=55] (1.4,0.3);
\draw[very thick, dotted, violet, ->] (0.8,0.6) to [bend left=45] (6.2,0.3);

\end{tikzpicture}}
\caption{$\exists$-linking $\Omega_{A}^{B}$}
\label{omega}
\end{subfigure}

\vspace{4mm}

\begin{subfigure}{\textwidth}
\adjustbox{scale=0.65,center}{$
A(E,\ell)\ = \quad
\begin{tikzpicture}[baseline=2ex]
\node(a) at (0,0) {$(\bigparr_{i}^{n}C_{i}^{\bot}\parr E)\otimes D[E]^{\bot}$};

\node(b) at (4,0) {$(\bigotimes_{i}^{n}C_{i}\otimes E^{\bot})\parr D[F]$};

\node(f) at (3.3,1.2) [draw,thick,minimum width=0.9cm,minimum height=0.25cm, rounded corners=3pt] {$\ell$};

\draw[thick] (-1,0.3) to [bend left=35] (3,0.3);
\draw[thick] (0.9,0.3) to [bend left=25] (5,0.3);
\draw[thick] (1.3,0.3) to [bend left=25] (f);
\draw[thick] (f) to [bend left=25] (5.4,0.3);

\draw[thick] (0.1,0.3) to [bend left=35] (4,0.3);

\end{tikzpicture}\qquad \qquad
A(\ell,F)\ = \quad
\begin{tikzpicture}[baseline=2ex]

\node(a) at (0,0) {$(\bigparr_{i}^{n}C_{i}^{\bot}\parr F)\otimes D[F]^{\bot}$};

\node(b) at (4,0) {$(\bigotimes_{i}^{n}C_{i}\otimes E^{\bot})\parr D[F]$};

\node(f) at (2.3,1.2) [draw,thick,minimum width=0.9cm,minimum height=0.25cm, rounded corners=3pt] {$\ell$};

\draw[thick] (-1,0.3) to [bend left=35] (3,0.3);
\draw[thick] (0.9,0.3) to [bend left=25] (5,0.3);
\draw[thick] (0.1,0.3) to [bend left=25] (f);
\draw[thick] (f) to [bend left=25] (4,0.3);

\draw[thick] (1.3,0.3) to [bend left=25] (5.4,0.3);


\end{tikzpicture}$
}
\caption{Functorial action of Yoneda formulas}
\label{precoend}
\end{subfigure}

\begin{subfigure}{\textwidth}
\adjustbox{scale=0.65,center}{
\begin{tikzpicture}[baseline=2ex]
\node(a) at (0,0) {$\exists X((\bigparr_{i}^{n}C_{i}^{\bot}\parr X)\otimes D[X]^{\bot})$};

\node(b) at (4,0) {$(\bigotimes_{i}^{n}C_{i}\otimes E^{\bot})\parr D[F]$};

\node(f) at (4.65,0.7) [draw,thick,minimum width=0.9cm,minimum height=0.25cm, rounded corners=3pt] {$\ell$};

\draw[thick] (-1,0.3) to [bend left=35] (3,0.3);
\draw[thick] (1.1,0.3) to [bend left=35] (5,0.3);
\draw[thick] (4,0.3) to [bend left=25] (f);
\draw[thick] (f) to [bend left=25] (5.3,0.3);

\draw[very thick, orange, dotted] (0.3,0.3) to [bend left=35] (1.5,0.3);
\draw[very thick, orange, dotted,->] (0.9,0.5) to [bend left=45] (4,0.3);

\end{tikzpicture}$\qquad\sim \qquad$
\begin{tikzpicture}[baseline=2ex]
\node(a) at (0,0) {$\exists X((\bigparr_{i}^{n}C_{i}^{\bot}\parr X)\otimes D[X]^{\bot})$};

\node(b) at (4,0) {$(\bigotimes_{i}^{n}C_{i}\otimes E^{\bot})\parr D[F]$};

\node(f) at (4.65,0.7) [draw,thick,minimum width=0.9cm,minimum height=0.25cm, rounded corners=3pt] {$\ell$};

\draw[thick] (-1,0.3) to [bend left=35] (3,0.3);
\draw[thick] (1.1,0.3) to [bend left=35] (5,0.3);
\draw[thick] (4,0.3) to [bend left=25] (f);
\draw[thick] (f) to [bend left=25] (5.3,0.3);

\draw[very thick, cyan, dotted] (0.3,0.3) to [bend left=35] (1.5,0.3);
\draw[very thick, cyan, dotted,->] (0.9,0.5) to [bend left=45] (5.3,0.3);

\end{tikzpicture}
}
%
%
%
%
%
%
%
%
\caption{$\exists$ is a co-wedge in $\BB L^{\exists}$}
\label{coendyes}
\end{subfigure}
\caption{Existential linkings and co-wedges.}
\end{figure}

\begin{example}
The ``Yoneda isomorphism''  holds in $\BB L^{\exists}$, as the composition $\ell_{Yo_{1}^{A}}\circ \ell_{Yo_{2}^{A}}$ reduces to $ID_{\forall XA}$ (up to rewitnessing).

\end{example}
By relying on the two Yoneda translations we now prove our main result.

\begin{theorem}\label{equiv}
$\pi$ and $\ell$ define an isomorphism of categories $\BB{G}_{\varepsilon}^{\C Y}\simeq \BB L^{\exists}$.
\end{theorem}
\begin{proof}
We will show that $\pi$ and $\ell$ are faithful functors inverse each other. 
To prove that $\pi$ is a faithful functor we must show that the assignment $\ell\mapsto \pi(\ell)$ yields an injective function $\BB L^{\exists}(A,B)\to \BB G_{\varepsilon}^{\C Y}(A,B)$. 
We claim that  $\ell\sim \ell' \Rightarrow \pi(\ell)\simeq_{\varepsilon}\pi(\ell')$: from $\ell\sim\ell'$ we deduce by lemma \ref{commuyoneda} $\ell_{\C Y}\sim_{lax}\ell'_{\C Y}$, hence, by proposition \ref{ionez} $a.$, $\B{Yon}(\pi(\ell))\sim_{lax} \B{Yon}(\pi(\ell'))$, and from the faithfulness of $\B{Yon}$ we can conclude $\pi(\ell)\simeq_{\varepsilon}\pi(\ell')$. This shows that $\pi$ is a function. Functoriality can be easily verified (by showing that $\pi$ maps identity linkings into identity $G$-nets and that it preserves composition). 
Injectivity is proved as follows: if $\pi(\ell)\simeq_{\varepsilon}\pi(\ell')$ then, by proposition \ref{ellpipi}, $\ell=\ell_{\pi(\ell)}\sim \ell_{\pi(\ell')}=\ell'$. 

To prove that $\ell$ is a faithful functor we must show that the assignment $\pi\mapsto \ell_{\pi}$ yields an injective function $ \BB G_{\varepsilon}^{\C Y}(A,B)\to \BB L^{\exists}(A,B)$. 
The functionality of $\ell$ follows from proposition \ref{ellpipi}. By construction it can be verified that the functor $\ell$ translates an identity $G$-net into an identity $\exists$-linking and that it preserves composition. Injectivity is proved as follows: if $\ell_{\pi}\sim \ell_{\pi'}$, then by lemma \ref{commuyoneda}, $(\ell_{\pi})_{\C Y} \sim_{lax}(\ell_{\pi'})_{\C Y}$, hence by proposition \ref{ionez} $b.$, $\B{Yon}(\pi)\sim_{lax}\B{Yon}(\pi')$ and from the faithfulness of $\B{Yon}$ we conclude $\pi\simeq_{\varepsilon}\pi'$.

Since $\ell=\ell_{\pi(\ell)}$, it remains to show that $\pi\simeq_{\varepsilon}\pi(\ell_{\pi})$. This follows from $\ell_{\pi}=\ell_{\pi(\ell_{\pi})}$ and the faithfulness of $\ell$.

\end{proof}

\begin{corollary}
For all $G$-nets $\pi,\pi'$ of conclusions $\Gamma$, $\pi\simeq_{\varepsilon} \pi'$ iff $\ell_{\pi}\sim \ell_{\pi'}$.
\end{corollary}

\section{Conclusions}

We provided a syntactic characterisation of the equational theory generated by ends/coends over \emph{Yoneda} formulas in
$\mathsf{MLL2}$. Our result relies on the simple structure of \emph{Yoneda} formulas (1 positive and 1 negative occurrence of quantified variables) and on the existence of a faithful translation from $\mathsf{MLL2}_{\D Y}$ to $\mathsf{MLL}$ with units.
It seems thus plausible that more sophisticated syntactic techniques are required to extend the characterisation to more expressive fragments of $\mathsf{MLL2}$.
In particular, while our result implies the decidability of the dinatural equivalence $\simeq_{\varepsilon}$ in $\mathsf{MLL2}_{\D Y}$, it is not known whether the theory $\simeq_{\varepsilon}$ is decidable over full $\mathsf{MLL2}$. 
However, keeping the Yoneda restriction, it can be expected that similar characterizations can be obtained for more expressive systems like $\mathsf{MELL2}$ (which is as expressive as System $\mathsf F$). 


Finally, it might be interesting to compare the theory $\simeq_{\varepsilon}$ with the equivalence arising from other models of $\mathsf{MLL2}$ investigated in the literature. For instance, while it is well-known that the coherent model of second order linear logic \cite{linear} is not dinatural (\cite{Fiore1996}), it can be easily seen that it satisfies the Yoneda isomorphism. Hence it can be conjectured that the model is injective (in the sense of \cite{Carvalho2012}) with respect to $\exists$-linkings for $\mathsf{MLL2}_{\C Y}$.

\bibliographystyle{eptcs}
\bibliography{YonedaFinal.bib}

\appendix

\section{$^{*}$-autonomous categories and coends}\label{appB}

We recall that a $^{*}$-autonomous category is a category $\BB C$ endowed with functors $\_\otimes\_: \BB C^{2}\to \BB C$ and $\_^{\bot}: \BB C^{op}\to \BB C$, an object $\B 1_{\BB C}$, the following natural isomorphisms:
$$
\begin{matrix}
\alpha_{a,b,c} \ :\ a\otimes (b\otimes c) \to (a\otimes b)\otimes c\\ 
\lambda_{a} \ : \ a\otimes \B 1_{\BB C}\to a\\
\rho_{a} \ : \ \B 1_{\BB C}\otimes a\to a \\
\sigma_{a,b} \ : \ a\otimes b\to b\otimes a \\
\end{matrix}
$$
and a natural bijection between $\BB C(a\otimes b,c)$ and $\BB C(a,  b^{\bot}\parr c)$, where $x\parr y= \BB C(x^{\bot},y)$, satisfying certain coherence conditions (that we omit here, see \cite{Barr1979}). In any $^{*}$-autonomous category $\BB C$ there is a natural isomorphism $A^{\bot\bot}\simeq A$. $\BB C$ is called \emph{strict} when this isomorphism is an identity.

For the definition of multivariant functors and dinatural transformations the reader can look at \cite{MacLane}. When  $F:(\BB C^{op}\otimes \BB C)^{n+1}\to \BB D$ and the values $a_{1},\dots, a_{n}\in Ob_{\BB C}$ are clear from the context, we will will often abbreviate $F( (a_{1},\dots, a_{n},a), (a_{1},\dots, a_{n},b))$ as $F(a,b)$. 

Given $\BB C$ $^{*}$-autonomous, for all $a\in Ob_{\BB C}$, there exist dinatural transformations $\hat{\B 1}_{x}:\B 1_{\BB C}\to x^{\bot}\parr x$ and $\hat{\bot}_{x}=\hat{\B 1}_{x}^{\bot}: x\otimes x^{\bot}\to \bot_{\BB C}$, where $\bot_{\BB C}:= \B 1_{\BB C}^{\bot}$. It is clear that such transformations exist in all dinatural model, according to definition \ref{dinamodel}.
%
%
%
%

Given categories $\BB C,\BB D$ and a multivariant functor $F:( \BB C^{op}\otimes \BB C)^{n+1}\to \BB D$, a \emph{wedge for $F$}\footnote{We give here a functorial definition of ends and coends which can be easily deduced from the usual definition (see \cite{MacLane}).} (dually, a \emph{co-wedge for $F$}, see \cite{MacLane}) is a pair $(C, \delta_{x_{1},\dots, x_{n},a})$ (resp. ($D, \omega_{x_{1},\dots, x_{n},a})$)\footnote{We will abbreviate $\delta_{x_{1},\dots, x_{n},a}$ and $\omega_{x_{1},\dots, x_{n},a}$ simply as $\delta_{a}$ and $\omega_{a}$, respectively.} made of a functor $C:(\BB C^{op}\otimes \BB C)^{n}\to \BB D$ and a  dinatural transformation $\delta_{a}: C\to F(a,a)$ (resp. $\omega_{a}: F(a,a)\to D $) natural in $x_{1},\dots, x_{n}$.
A wedge (resp. a co-wedge) for $F$ is an \emph{end} (resp. a \emph{coend}) when the dinatural transformation $\delta_{a}$ (resp. $\omega_{a}$) is \emph{universal}.
 This means that for any functor $G:(\BB C^{op}\otimes \BB C)^{n}\to \BB D$ and dinatural transformation $\theta_{a}: G\to F(a,a)$ (resp. $\theta_{a}:F(a,a)\to G$) there exists a unique natural transformation $h: G\to \int_{x}F(x,x)$ (resp. $k: \int^{x}F(x,x)\to G$) such that the following diagrams commute for all $f\in \BB C(a,b)$:

\adjustbox{scale=0.8,center}{
\begin{tikzcd}
G \ar[rrd,"{\theta_{a}}", bend left=15] \ar[ddr, "{\theta_{b}}", bend right=15] \ar[dashed]{rd}{h}  &    &   \\   
  &  \int_{x} F(a,a) \ar[r, "\delta_{a}"] \ar[d, "{\delta_{b}}"] & F(a,a) \ar[d, "{F(a{,}f)}"] \\
  &  F(b,b) \ar[r,"{F(f{,}b)}" below]  &  F(a,b)
  \end{tikzcd}
  $\qquad$
\begin{tikzcd}
  F(b,a) \ar[r,"{F(f{,}a)}"] \ar[d, "{F(b{,}f)}"]  & F(a,a) \ar[rdd, "{\theta_{a}}", bend left=15] \ar[d, "{\omega_{a}}"]&  \\
  F(b,b) \ar[r, "{\omega_{b}}"] \ar[rrd, "{\theta_{b}}", bend right=15]&  \int^{x} F(x,x) \ar[dashed]{rd}{k} & \\
    &    &  G
   \end{tikzcd}}
   
\noindent 
Duality yields $\int_{x}F=(\int^{x}F^{\bot})^{\bot}$, $\int^{x}F=(\int_{x}F^{\bot})^{\bot}$ and $\delta_{a}=\omega_{a}^{\bot}$, $\omega_{a}=\delta_{a}^{\bot}$.


We recall some basic facts about coends (see \cite{MacLane, Loregian}):

\begin{itemize}

\item Commutation with $\parr/\otimes$:

\begin{eqnarray}\label{commut}
\int_{x}(F\parr G(x,x))\simeq G\parr \int_{x}G(x,x) \\
\int^{x}(F\otimes G(x,x))\simeq F\otimes \int^{x}G(x,x)
\end{eqnarray}

\item ``Fubini'' theorem: 
\begin{eqnarray}\label{fubini}
\int_{x}\int_{y}F\simeq\int_{y}\int_{x}F \\
\int^{x}\int^{y}F\simeq\int^{y}\int^{x}F
\end{eqnarray}

\item Commutation of $\int_{x}/\int^{x}$ and $\parr$: given a functor $F$ and a multivariant functor $G(x,y)$, there exist natural transformations

\begin{eqnarray}\label{coends2}
\mu:\int_{x}(F\parr G(x,x))\to F\parr \int_{x}G(x,x)\\
\nu:\int^{x}(F\parr G(x,x))\to F\parr \int^{x}G(x,x)
\end{eqnarray}

\end{itemize}

In a dinatural model (def. \ref{dinamodel}) one considers \emph{relativized} ends and coends, that is, wedges/co-wedges which are universal with respect to a certain class of (composable) dinatural transformations. 
All facts above about ends and coends can be straightforwardly adapted to relativized ends and coends.

\section{Hughes sequentialization theorem}\label{AppC}

We adapt the sequentialization algorithm for unification nets in \cite{Hughes2018} to $G$-nets. This algorithm is based on the translation of a unification net into a $\mathsf{MLL}^{-}$ proof net  (where $\mathsf{MLL}^{-}$ indicates $\mathsf{MLL}$ without units), called the \emph{frame}, by a suitable encoding of jumps. The reconstruction of a sequent calculus derivation exploits then the usual splitting property of $\mathsf{MLL}^{-}$ proof nets. This construction can be straightforwardly adapted to $G$-nets, by translating a cut-free $G$-proof structures into $\mathsf{MLL}^{-}$ proof-structures as follows:\\
$(1)$ \ \emph{Encode every jump from a $\forall$ to an $\exists$ as a new link}: for each such jump between formulas $\forall XA$ and $\exists YB$, let $Z$ be a fresh variable. Replace $\exists YB$ by $Z\otimes \exists YB$ and $\forall XA$ by $Z^{\bot}\parr \forall XA$; \\
$(2)$\ \emph{Delete quantifiers}. After (1) replace every formula $\forall XA$ by $A$ and every formula $\exists XA$, with premiss $A[B/X]$, by $A[B/X]$. \\
We let $\pi_{m}$, the \emph{frame of $\pi$}, be the $\mathsf{MLL}^{-}$ proof-structure obtained. The following two lemmas are as in \cite{Hughes2018}.

\begin{lemma}
If $\pi$ is a $G$-net, $\pi_{m}$ is a proof net.
\end{lemma}
\begin{lemma}\label{33}
No $\otimes$ added during the construction of $\pi_{m}$ splits.
\end{lemma}

We can now use $\pi_{m}$ to find splitting tensors in $\pi$, yielding the following:
\begin{theorem}[sequentialization]
If $\pi$ is a $G$-net, then $\pi$ is the translation of some sequent calculus derivation.
\end{theorem}
\begin{proof}
The sequentialization algorithm for a $G$-net $\pi$ is as follows: \begin{enumerate}
\item[1.] Start by eliminating negative links, i.e. $\parr, \forall$ links; in other words, for any link of conclusion $A\parr B$ (resp. $\forall XA$), let $\pi'$ be the $G$-net obtained by deleting the $\parr$ (resp. $\forall$) link. By induction hypothesis $\pi'$ is sequentializable, yielding a derivation of $\Gamma-\{A\parr B \}, A, B$ (resp. $\Gamma-\{\forall XA\}, A$), from which a derivation of $\Gamma$ can be obtained by a $\parr$-rule (resp. by a $\forall$-rule -  we are here supposing that $\Gamma, \forall XA$ is clean, so $X$ does not occur free in $\Gamma$).

\item[2.] If, after 1, there are $\exists$-links with no incoming jumps, eliminate them;  in other words, for any such link of conclusion $\exists XA$, let $\pi'$ be the $G$-net obtained by deleting the link. By induction hypothesis $\pi'$ is sequentializable, yielding a derivation of $\Gamma-\{\exists XA \}, A[B/X]$, for some formula $B$, from which a derivation of $\Gamma$ can be obtained by a $\exists$-rule.

\item[3.] After 2 all non-axiom links are either $\otimes$ or $\exists$ with incoming jumps. If there is none we are done. Otherwise $\pi_{m}$ has only $\otimes$-links, so one must be splitting, and by lemma \ref{33} it corresponds to a splitting $\otimes$ in $\pi$. 
By deleting this link we obtain two $G$-nets $\pi_{1}$, $\pi_{2}$ yielding, by induction hypothesis, two derivations of conclusions, respectively, $\Gamma_{1}, A$ and $\Gamma_{2}, B$, where $\Gamma=\Gamma_{1},\Gamma_{2},A\otimes B$. Now, a derivation of $\Gamma$ is obtained by a $\otimes$-rule.

\end{enumerate}

\end{proof}

\section{Proof of lemma \ref{ready}}\label{C}

To prove lemma \ref{ready} (the ``ready lemma'') we use the following facts, which can be easily established by looking at $\pi(\ell)$:
\begin{lemma}[$\bot^{\exists}$-moves]\label{moves}
($i.$)  \ If $A_{c}=\bot^{\exists}$ and $W(c)=B$ occurs in a cut-formula $B\otimes B^{\bot}$, then $c$ can be rewired on $B^{\bot}$.\\
($ii.$) \ If $A_{c}=\bot^{\exists}$ and $W(c)=B$, then $c$ can be rewired on any subformula of $B$.\\
($iii.$) \ If $A_{c}=\bot^{\exists}$ and $W(c)=X$ is the conclusion of an axiom link of conclusions $X,X^{\bot}$, then $c$ can be rewired on $X^{\bot}$.
\end{lemma}
From lemma \ref{moves} we deduce:
\begin{proposition}\label{moves2}
If for all $c\in \Gamma^{\exists}$, $A_{c}=\bot^{\exists}$, then $\ell$ is equivalent to a ready $\exists$-linking.
\end{proposition}
\begin{proof}
For any cut formula $B\otimes B^{\bot}$, there is at least an axiom link going outside the tree of $B$ and $B^{\bot}$, otherwise both $B$ and $B^{\bot}$ would be provable. Hence, if $W(c)$ is in the tree of a cut formula $B\otimes B^{\bot}$, by lemma \ref{moves} it can be rewitnessed upwards so to pass through an axiom links moving outside the cut.

\end{proof}

\begin{figure}
\adjustbox{scale=0.8, center}{$
\begin{matrix}
\ell=\qquad
 \begin{tikzpicture}[baseline=-12ex]
\draw (0,0) ellipse (0.9cm and 0.5cm);

\node(p1) at (0,0) {$\pi_{1}$};
\node(p2) at (1.3,-1.2) {$\pi_{2}$};

\draw (1.3,-1.2) ellipse (0.6cm and 0.5cm);

\node(c) at (-0.5,-0.7) {$C$};
\node(a) at (0.5,-0.7) {$A$};

\node(p) at (1,-2.8) {$\parr$};

\node(x) at (2,-2.1) {$\B X$};

\node(t) at (2,-3.5) {$\otimes$};
\node(d) at (3.3,-2.4) {$D[\B X^{\bot}]$};

\node(aa) at (4,-2.8) {$\B Y$};
\node(aaa) at (5,-2.8) {$\B Y^{\bot}$};

\node(tt) at (4.5,-3.5) {$\otimes$};

\node(e) at (2,-4.2) {$\exists$};

\node(ee) at (4.5,-4.2) {$\bot^{\exists}$};

\draw[thick, bend right=15] (c) to (p);
\draw[thick] (t) to (p);
\draw[thick] (p) to (x);

\draw[thick] (t) to (e);
\draw[thick] (d) to (t);
\draw[thick] (aa) to (tt);
\draw[thick] (aaa) to (tt);
\draw[thick] (ee) to (tt);

\draw[very thick, dotted, orange] (x) to [bend left=55] (d);
\draw[very thick, dotted, orange, ->] (2.7,-1.7) to [bend right=55] (a);

\draw[very thick, dotted, violet] (aa) to [bend left=55] (aaa);
\draw[very thick, dotted, violet, ->] (4.5,-2.4) to [bend right=75] (d);

\end{tikzpicture}
& &
\ell'= \begin{tikzpicture}[baseline=-12ex]
\draw (0,0) ellipse (0.9cm and 0.5cm);

\node(p1) at (0,0) {$\pi_{1}$};
\node(p2) at (1.3,-1.2) {$\pi_{2}$};

\draw (1.3,-1.2) ellipse (0.6cm and 0.5cm);

\node(c) at (-0.5,-0.7) {$C$};
\node(a) at (0.5,-0.7) {$A$};

\node(p) at (1,-2.8) {$\parr$};

\node(x) at (2,-2.1) {$\B X$};

\node(t) at (2,-3.5) {$\otimes$};
\node(d) at (3.3,-2.4) {$D[\B X^{\bot}]$};

\node(aa) at (4,-2.8) {$\B Y$};
\node(aaa) at (5,-2.8) {$\B Y^{\bot}$};

\node(tt) at (4.5,-3.5) {$\otimes$};

\node(e) at (2,-4.2) {$\exists$};

\node(ee) at (4.5,-4.2) {$\bot^{\exists}$};

\draw[thick, bend right=15] (c) to (p);
\draw[thick] (t) to (p);
\draw[thick] (p) to (x);

\draw[thick] (t) to (e);
\draw[thick] (d) to (t);
\draw[thick] (aa) to (tt);
\draw[thick] (aaa) to (tt);
\draw[thick] (ee) to (tt);

\draw[very thick, dotted, orange] (x) to [bend left=55] (d);
\draw[very thick, dotted, orange, ->] (2.7,-1.7) to [bend left=75] (ee);

\draw[very thick, dotted, violet] (aa) to [bend left=55] (aaa);
\draw[very thick, dotted, violet, ->] (4.5,-2.4) to [bend right=55] (a);

\end{tikzpicture} 
\end{matrix}$}
\caption{From $\ell$ to $\ell'$ by two rewitnessing moves. }
\label{uff1}
\end{figure}
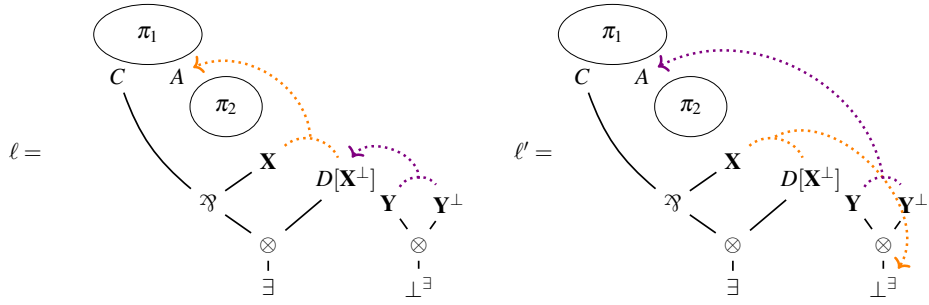

\begin{figure}
\adjustbox{scale=0.8, center}{$
\begin{matrix}
\pi(\ell)= \qquad \begin{tikzpicture}[baseline=-6ex]
\draw (0,0) ellipse (0.9cm and 0.5cm);

\node(p1) at (0,0) {$\pi_{1}$};
\node(p2) at (5.5,-2.1) {$\pi_{2}$};

\node(ax) at (5,-1.4) {$A$};

\draw (5.5,-2.1) ellipse (0.6cm and 0.5cm);

\node(c) at (-0.5,-0.7) {$C$};
\node(a) at (0.5,-0.7) {$A$};
\node(p) at (0,-1.4) {$\parr$};
\node(t) at (1,-2.1) {$\otimes$};
\node(d) at (2,-1.4) {$D[A^{\bot}]$};

\node(aa) at (3,-1.4) {$A$};
\node(aaa) at (4,-1.4) {$A^{\bot}$};
\node(aaaa) at (2,0) {$A^{\bot}$};

\draw[thick, bend left=55] (aaa) to (ax);

\node(tt) at (3.5,-2.1) {$\otimes$};

\node(e) at (1,-2.8) {$\exists$};

\node(ee) at (3.5,-2.8) {$\bot^{\exists}$};

\draw (2,-1.1) to (1.3,-0.3) to (2.7,-0.3) to (2,-1.1);

\draw[thick] (c) to (p);
\draw[thick] (a) to (p);
\draw[thick] (t) to (p);
\draw[thick] (t) to (e);
\draw[thick] (d) to (t);
\draw[thick] (aa) to (tt);
\draw[thick] (aaa) to (tt);
\draw[thick] (ee) to (tt);

\draw[thick, bend left=85] (aaaa) to (aa);

\end{tikzpicture}
&  &
\pi(\ell')=\qquad
\begin{tikzpicture}[baseline=-12ex]

\draw (0,0) ellipse (0.9cm and 0.5cm);

\draw (3.5,-0.5) ellipse (0.6cm and 0.5cm);

\node(p1) at (0,0) {$\pi_{1}$};
\node(p2) at (3.5,-0.5) {$\pi_{2}$};

\node(ax) at (3,0.2) {$A$};

\node(c) at (-0.5,-0.7) {$C$};
\node(a) at (0.5,-0.7) {$A$};

\node(aaaa) at (1.5,-0.7) {$A^{\bot}$};
\node(tt) at (1,-1.4) {$\otimes$};
\node(ee) at (1,-2.1) {$\bot^{\exists}$};

\node(p) at (0.25,-2.8) {$\parr$};
\node(t) at (1,-3.5) {$\otimes$};
\node(e) at (1,-4.2) {$\exists$};

\node(d) at (2,-2.8) {$D[\B 1^{\exists}]$};

\node(aa) at (2,-1.4) {$\B 1^{\exists}$};
\node(aaa) at (4,-2.8) {$\bot^{\exists}$};

\draw[thick, bend left=45] (aaaa) to (ax);

\draw (2,-2.5) to (1.3,-1.7) to (2.7,-1.7) to (2,-2.5);

\draw[thick, bend right=14] (c) to (p);
\draw[thick] (ee) to (p);
\draw[thick] (t) to (p);
\draw[thick] (t) to (e);
\draw[thick] (d) to (t);
\draw[thick] (a) to (tt);
\draw[thick] (aaaa) to (tt);
\draw[thick] (ee) to (tt);

\draw[thick, bend right=65] (aaa) to (aa);

\end{tikzpicture}
\end{matrix}$}
\caption{$\pi(\ell) $ and $\pi(\ell')$ are both correct.}
\label{uff}
\end{figure}
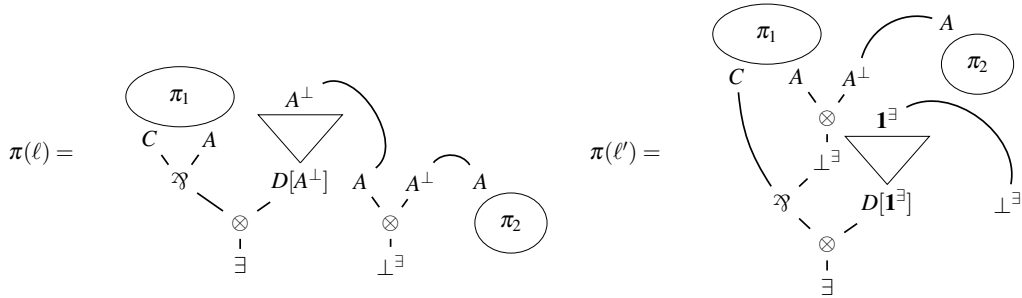

\begin{proof}[Proof of lemma \ref{ready}]
Given $\ell=(E,W)$, we will first construct an $\exists$-linking $\ell^{*}=(E,W^{*})$ such that $\ell\sim \ell^{*}$ and for all formula $A$ occurring in a cut, $(W^{*})^{-1}(A)$ is either empty of contains a formula of the form $\bot^{\exists}$. From this we can conclude then by applying proposition \ref{moves}.

Let $c=(\B X,\B X^{\bot})\in \Gamma^{\exists}$ be such that $A_{c}$ is not of the form $\bot^{\exists}$ and $W(c)=A$ occurs in a cut.
We can suppose that $W^{-1}(\B X^{\bot})$ contains $c'=(\B Y, \B Y')$ such that $A_{c'}=\bot^{\exists}$ is a conclusion of $\ell$ and such that $W^{-1}(\bot^{\exists})=\emptyset$: if it is not the case then we can add the formula $\bot^{\exists}$ to the conclusions of $\ell$ and set $W(c')=\B X^{\bot}$, as this preserves correctness and does not alter equivalence questions because of the isomorphism between the conclusions $\Gamma$ of $\ell$ and $\Gamma\parr \bot^{\exists}$. We let then $W'$ be like $W$ but for $W'(c)=\bot^{\exists}$ and $W'(c')=A$ (as illustrated in figure \ref{uff1}). $W'$ is obtained from $W$ by a rewitnessing move of type (2) (switching $W(c)$ and $W(c')$ so that $c$ is sent to $\B Y$ and $c'$ to $A$) and a rewitnessing move of type $(1)$ (moving $c$ from $\B Y$ to $\bot^{\exists}$). We must then show that $\ell'=(E,W')$ is correct, so that $\ell\sim \ell'$. This follows by remarking that the first rewitnessing move does not change $\pi(\ell)$ and that the second rewitnessing move transforms $\pi(\ell)$ into $\pi(\ell')$ (as illustrated in fig. \ref{uff}), preserving correctness, as it can be seen by inspecting paths in both graphs. By applying this operation to all co-edges $c$ such that $A_{c}\neq \bot^{\exists}$ we obtain the desired $\exists$-linking $\ell^{*}\sim \ell$.

\end{proof}

\end{document}